\documentclass[10pt,a4paper]{elsarticle}

\usepackage{lineno}
\modulolinenumbers[5]

\usepackage{geometry}
\usepackage{appendix}

\usepackage{pdfpages}
\usepackage{graphicx}
\graphicspath{{figs/}}
\usepackage[section]{placeins}
\makeatletter
\AtBeginDocument{%
  \expandafter\renewcommand\expandafter\subsection\expandafter{%
    \expandafter\@fb@secFB\subsection
  }%
}
\makeatother

\usepackage{amsmath, amssymb, amsthm}
\usepackage{mathtools} 
\usepackage[linesnumbered,ruled,vlined]{algorithm2e}

\usepackage{enumerate} 

\newtheorem{proposition}{Proposition}[section]

\theoremstyle{definition}
\newtheorem{definition}{Definition}[section]
\theoremstyle{definition}
\newtheorem{exmp}{Example}[section]
\theoremstyle{remark}
\newtheorem{assumption}{Assumption}[section]
\theoremstyle{remark}
\newtheorem{remark}{Remark}[section]
\theoremstyle{remark}

\usepackage{bm}
\bmdefine{\bA}{A}
\bmdefine{\bB}{B}
\bmdefine{\bC}{C}
\bmdefine{\bc}{c}
\bmdefine{\bD}{D}
\bmdefine{\bd}{d}
\bmdefine{\be}{e}
\bmdefine{\bF}{F}
\bmdefine{\bG}{G}
\bmdefine{\bH}{H}
\bmdefine{\bI}{I}
\bmdefine{\bK}{K}
\bmdefine{\bV}{V}
\bmdefine{\bM}{M}
\bmdefine{\bf}{f}
\bmdefine{\bq}{q}
\bmdefine{\bT}{T}
\bmdefine{\bt}{t}
\bmdefine{\bu}{u}
\bmdefine{\bU}{U}
\bmdefine{\bx}{x}
\bmdefine{\bX}{X}
\bmdefine{\by}{y}
\bmdefine{\bY}{Y}
\bmdefine{\bomega}{\omega}
\bmdefine{\bupsilon}{\upsilon}
\bmdefine{\bP}{P}
\bmdefine{\bQ}{Q}
\bmdefine{\bS}{S}
\bmdefine{\bR}{R}
\bmdefine{\bJ}{J}
\bmdefine{\bV}{V}
\bmdefine{\bW}{W}
\bmdefine{\bO}{O}
\bmdefine{\bzeta}{\zeta}
\bmdefine{\beeta}{\eta}
\bmdefine{\bLambda}{\Lambda}
\bmdefine{\bGamma}{\Gamma}
\bmdefine{\balpha}{\alpha}
\bmdefine{\bbeta}{\beta}
\bmdefine{\bphi}{\phi}
\bmdefine{\bpsi}{\psi}
\bmdefine{\bSigma}{\Sigma}
\bmdefine{\bzero}{0}

\makeatletter
\def\els@aparagraph[#1]#2{\elsparagraph[#1]{#2}}
\def\els@bparagraph#1{\elsparagraph*{#1}}
\makeatother

\makeatletter
\def\ps@pprintTitle{%
    \let\@oddhead\@empty
    \let\@evenhead\@empty
    \let\@oddfoot\@empty
    \let\@evenfoot\@oddfoot
}
\makeatother

\journal{Elsevier}

\begin{document}

\begin{frontmatter}

\title{Blind Identification of State-Space Models in Physical Coordinates}

\author{Runzhe Han\corref{mycorrespondingauthor}}
\ead{runzhe.han@alumni.tu-clausthal.de}
\cortext[mycorrespondingauthor]{Corresponding author}
\author{Christian Bohn\corref{}}
\author{Georg Bauer\corref{}}
\address{Institut f{\"u}r Elektrische Informationstechnik (IEI) \\
Technische Universit{\"a}t Clausthal \\
Leibnizstr. 28, 38678 Clausthal-Zellerfeld, Germany}

\begin{abstract}
Blind identification is popular for modeling a system without the input information, such as in the research areas of structural health monitoring and audio signal processing. Existing blind identification methods have both advantages and disadvantages, in this paper, we briefly outline current methods and propose a novel blind identification method for identifying state-space models in physical coordinates. The idea behind this proposed method is first to regard the collected input data of a state-space model as a part of a periodic signal sequence, and then transform the state-space model with input and output into a model without input by augmenting the state-space model with the input model (which is a periodic signal model), and afterwards use merely the output information to identify a state-space model up to a similarity transformation, and finally derive the state-space model in physical coordinates by using a unique similarity transformation. With the above idea, physical parameters and modal parameters of a state-space system can be obtained. Both numerical and practical examples were used to validate the proposed method. The result showed the effectiveness of the novel blind identification method.
\end{abstract}

\begin{keyword}
State-space model in physical coordinates \sep periodic signal model \sep subspace model identification \sep output-only identification \sep modal parameter estimation \sep input estimation \sep structural health monitoring
\end{keyword}

\end{frontmatter}


\section{Introduction}\label{s1}
System identification is a process to derive the mathematical model of a static or a dynamic system using experimental data. The problem of system identification in control community has been explored in both time domain (e.g., time-series models and state-space models) and frequency domain (e.g., frequency response functions) since $1960$s \cite{Gevers2006}. Among all mathematical model structures, state-space models are popular in the fields of mechanical engineering (e.g., structural health monitoring) and control engineering (e.g., model predictive control). As a typical example, in the research area of structural health monitoring, the linear finite element technique (for the nonlinear one, see \cite{Wriggers2008}) is used almost exclusively for constructing analytical models in the form of second-order differential equations \cite{Agbabian1991, Juang1994}. These analytical models can be reformulated as a first-order system of differential equations in a number of ways, such as a state-space model, thus with experimental data modal parameters (e.g., natural frequency, natural mode, and damping ratio) can be derived from the state-space model identified by various existing state-space system identification methods, such as observer-based methods (see, e.g., \cite{Juang1994}, and the references therein) or subspace-based state-space system identification ($4$SID, pronounced ``forsid'') (see, e.g., \cite{Peeters1999, Basseville2001, Peeters2001, Mevel2006}, and the references therein). A summary of $4$SID methods can be referred to \cite{Overschee1996}, \cite{Ljung1999}, and \cite{Katayama2006}. Furthermore, if the identified state-space model is in physical coordinates, of which the specific meaning can be found in \cite{Phan2004}, both system parameters (also called physical parameters, e.g., stiffness and damping) and modal parameters of the system of interest can be obtained, as evidenced by the works of \cite{Angelis2002, Lus2003, Phan2004, Kim2012a, Kim2012b, Park2017} (and the related references therein), in which both experimental input and output data were used for identification. It should be noted that both system parameters and modal parameters can be used for the structural health monitoring purpose, but the former ones have shown to have higher sensitivity with respect to structural damages and less susceptibility to temporary environmental effects \cite{Ghobadi2017}. Therefore, the estimation of physical parameters is more desirable in structural health monitoring applications. Additionally, if physical parameters are obtained, modal parameters can be naturally derived then.

However, for state-space system in physical coordinates identification, one case that the information of input sources can be unknown or obtained under adverse conditions (such that measurement sensors can have a short lifecycle) is usually encountered in practice, so output-only identification methods are preferred. Generally, there are two frameworks for dealing with the identification of physical parameters and modal parameters without input information: i) output-only system identification (see, e.g., \cite{Caicedo2004}, and \cite{Alicioglu2008} without physical parameter identification), and ii) blind system identification (which can simultaneously identify system parameters and modal parameters as well as the history of unknown input excitations) (see, e.g., \cite{Wang1997, Shi2016, Hussein2016, Ghobadi2017}). Output-only system identification usually assumes the unknown inputs have a certain known characteristics such as white noise process or general stationary random process, which can be violated in practical applications. As a comparison, blind identification avoids this assumption and can even estimate the unknown inputs.

In \cite{Bohn2004} and \cite{Han2018a}, based on frequency-modulated signal modeling, an augmented model with only output can be derived. Inspired by the formation of the output-only augmented model, in this paper, a novel blind identification method is proposed to identify both system parameters and modal parameters of a state-space model in physical coordinates as well as the input. The proposed blind identification method can be seen as an extension of the identification method (named AOSID) proposed in \cite{Ghobadi2017} (which was validated as an effective method), but it is mainly different in the following three aspects:

\begin{enumerate}[{(1)}]
\item
The conventional stochastic subspace identification method (see Chapter $3$ in \cite{Overschee1996}) is used to identify the augmented model, instead of using the eigensystem realization algorithm (ERA) (see \cite{Juang1994}) in AOSID.
\item
The assumption that there does not exist frequency overlapping between natural frequencies of the physical system and exciting frequencies is not necessary.
\item
The singular value decomposition (SVD) technique becomes unnecessary for the order determination of the signal model.
\end{enumerate}

There are other contributions of the proposed method which will be illustrated in the following sections.

The rest of the paper is organized as follows. Section \ref{s2} formulates the identification problem of state-space models in physical coordinates. Section \ref{s3} describes both periodic signal modeling and augmented model derivation, followed by Section \ref{s4}, in which a novel blind identification method is proposed. Section \ref{s5} and Section \ref{s6} demonstrate the performance of the proposed identification method using both numerical and practical examples. Finally some conclusions are made in Section \ref{s7}. It should be noted that in this paper, the proposed blind identification method is an offline identification method.

\section{Problem formulation}\label{s2}
In the structures field, finite element modeling is of prime importance. The fundamental equation describing the dynamic behavior of a structure can be discretized by finite element, each element can stand for one degree of freedom. The second-order equation of motion of an n-degree-of-freedom mechanical system with inputs has the general form

\begin{equation}\label{eq1}
\bM\ddot{\bq}(t)+\bD\dot{\bq}(t)+\bK\bq(t)=\bB\bu(t)=\bf_{\rm e}(t),
\end{equation}
where $\bq(t)\in\mathbb{R}^n$ is the absolute displacement, $\bM\in\mathbb{R}^{n\times n}$, $\bD\in\mathbb{R}^{n\times n}$, and $\bK\in\mathbb{R}^{n\times n}$ are mass, damping, and stiffness matrices, respectively. The matrix $\bB\in\mathbb{R}^{n\times r}$ is the input influence matrix of the second-order equation, where typically $r\leqslant n$. If the system has a full set of actuators, the matrix $\bB$ will be square and full rank. The input $\bu(t)$ is a band-limited signal. We call $\bf_{\rm e}(t)$ effective input.

For the system (\ref{eq1}), if there is a full set of $n$ sensors measuring displacements, velocities, or accelerations, then we can define the measurements as
\begin{equation}\label{eq2}
\by(t)=\bC_{\rm{p}}\bq(t)+\bC_{\rm{v}}\dot{\bq}(t)+\bC_{\rm{ac}}\ddot{\bq}(t)+\be_{\rm m}(t),
\end{equation}
where $\bC_{\rm{p}}$, $\bC_{\rm{v}}$, and $\bC_{\rm{ac}}\in\mathbb{R}^{m\times n}$ are measurement sensitivity matrices which are known from the type and placement of the sensors in most mechanical systems.

The following assumptions are made for the system (\ref{eq1}) in this paper:

\begin{assumption}
A full-state observation is assumed through this paper, which means that the number of sensors are exactly the same as the number of degree of freedom (DOF) of the system (\ref{eq1}), i.e., $m=n$. $\be_{\rm m}(t)$ denotes the measurement noise.
\end{assumption}

\begin{assumption}\label{assum2}
Only the acceleration sensor is used, i.e., $\bC_{\rm{p}}=\bC_{\rm{v}}=\bzero\in\mathbb{R}^{n\times n}$.
\end{assumption}

According to the modern control theory and the above assumptions, the system described by Equations (\ref{eq1}) and (\ref{eq2}) can be reformulated in the form of a first-order state-space model in physical coordinates
\begin{equation}\label{eq3}
    \begin{dcases}
        \dot\bx_{\rm{s}}(t)=\bA_{\rm{s}}\bx_{\rm{s}}(t)+\bB_{\rm{s}}\bu(t), \\
        \by(t)=\bC_{\rm{s}}\bx_{\rm{s}}(t)+\bD_{\rm{s}}\bu(t)+\be_{\rm m}(t), \\
    \end{dcases}
\end{equation}
where
$$\bx_{\rm{s}}=
\begin{pmatrix}
\bq(t) \\
\dot\bq(t) \\
\end{pmatrix},
\bA_{\rm{s}}=
\begin{pmatrix}
\bzero & \bI \\
-\mathcal{K} & -\mathcal{D} \\
\end{pmatrix},
\bB_{\rm{s}}=
\begin{pmatrix}
\bzero \\
\mathcal{B} \\
\end{pmatrix},$$
$$\bC_{\rm{s}}=
\begin{pmatrix}
\bC_{\rm{p}}-\bC_{\rm{ac}}\mathcal{K} & \bC_{\rm{v}}-\bC_{\rm{ac}}\mathcal{D} \\
\end{pmatrix},
\bD_{\rm{s}}=\bC_{\rm{ac}}\mathcal{B},$$
and the matrix $\bI$ denotes an identity matrix, $\mathcal{K}=\bM^{-1}\bK$, $\mathcal{D}=\bM^{-1}\bD$, and $\mathcal{B}=\bM^{-1}\bB$. (We call $\mathcal{K}$, $\mathcal{D}$, and $\mathcal{B}$ normalized stiffness, damping, and input matrices with respect to the mass matrix $\bM$, respectively.)

In this paper, we attempt to identify physical parameters ($\mathcal{K}$ and $\mathcal{D}$), modal parameters (damping ratios and natural frequencies) of the system (\ref{eq1}), and the effective input $\bf_e(t)$ by only using the output signal $\by(t)$ offline. So the problem of this paper is formulated.

\section{Periodic signal modeling and augmentation}\label{s3}
In this section, the modeling of a periodic signal is first described. Specifically both continuous-time periodic signal model and discrete-time periodic signal model are presented, and then the process of augmenting a transfer path with the input signal model is illustrated. The reason why augmentation is needed will be shown in Section \ref{s4}.

\subsection{Periodic signal modeling}\label{periodic signal model}
A continuous-time period signal $u_{\rm{T}}(t)\in\mathbb{R}$ can be represented as a Fourier series with amplitudes $a_{i}$ and phases $\alpha_{i}$, that is
\begin{equation}\label{eq4}
u_{\rm{T}}(t)=a_{0}+\sum_{i=1}^{\infty}a_{i}\sin(i2\pi f_{\rm{T}}t+\alpha_{i}),
\end{equation}
where $f_{T}$ is the fundamental frequency.

However, if the periodic signal $u_{\rm{T}}$ is band-limited, it can be approximately represented as a Fourier series with limited terms, that is
\begin{equation}\label{eq5}
u_{\rm{T}}(t)=a_{0}+\sum_{i=1}^{n_{\rm{a}}}a_{i}\sin(i2\pi f_{\rm{T}}t+\alpha_{i})+e_{\rm o}(t),
\end{equation}
where $n_{\rm{a}}$ is the considered number of the harmonics which depends on both fundamental frequency and maximum frequency of $u_{\rm{T}}(t)$. $e_{\rm o}(t)$ denotes the part of $u_{\rm{T}}(t)$ which is not considered.

Based on the above periodic signal introduction, in the following of this subsection, the periodic signal modeling is described in both continuous time and discrete time.

\subsubsection{Continuous-time periodic signal modeling}\label{ctpm}
According to (\ref{eq5}) and \cite{Bohn2005}, $u_{\rm{T}}(t)$ can be modeled as the output of an autonomous state-space model, that is
\begin{equation}\label{eq6}
    \begin{dcases}
        \dot\bx_{u_{\rm{T}}}(t)=\bA_{u_{\rm{T}}}\bx_{u_{\rm{T}}}(t), \\
        u_{\rm{T}}(t)=\bC_{u_{\rm{T}}}\bx_{u_{\rm{T}}}(t)+e_{\rm o}(t), \\
    \end{dcases}
\end{equation}
where $\bx_{u_{\rm{T}}}(t)\in\mathbb{R}^{2n_{\rm{a}}-1}$ is the state vector.

The matrices $\bA_{u_{\rm{T}}}$ and $\bC_{u_{\rm{T}}}$ of the system (\ref{eq6}) are given as
\begin{equation}\label{eq7}
\bA_{u_{\rm{T}}}=
\begin{pmatrix}
0 & \bzero & \cdots & \bzero \\
\bzero & \bA_{u_{\rm{T}},\rm{1}} & \ddots & \vdots \\
\vdots & \ddots & \ddots & \bzero \\
\bzero & \cdots & \bzero & \bA_{u_{\rm{T}},n_{\rm{a}}-1} \\
\end{pmatrix}
\end{equation}
and
\begin{equation}\label{eq8}
\bC_{u_{\rm{T}}}=
\begin{pmatrix}
1 & \bC_{u_{\rm{T}},\rm{1}} & \cdots & \bC_{u_{\rm{T}},n_{\rm{a}}-1} \\
\end{pmatrix}.
\end{equation}

The individual block entries in these block matrices can be represented as
\begin{equation}\label{eq9}
\bA_{u_{\rm{T}},i}=
\begin{pmatrix}
0 & i2\pi f_{\rm{T}} \\
-i2\pi f_{\rm{T}} & 0 \\
\end{pmatrix}
\end{equation}
and
\begin{equation}\label{eq10}
\bC_{u_{\rm{T}},i}=
\begin{pmatrix}
1 & 0 \\
\end{pmatrix}.
\end{equation}

It should be noted that the matrix $\bA_{u_{\rm{T}}}$ is not stable which contains the eigenvalue $0$.

\subsubsection{Discrete-time periodic signal modeling}\label{dtpm}
Similar to the continuous-time periodic signal $u_{\rm{T}}(t)$, its discrete-time version $u_{\rm{T}}(k)$ can be also modeled as the output of an autonomous state-space model \cite{Bohn2005}, that is
\begin{equation}\label{eq11}
    \begin{dcases}
        \dot\bx_{u_{\rm{T}}}(k)=\bA_{u_{\rm{T}}}^{\rm{d}}\bx_{u_{\rm{T}}}(k), \\
        u_{\rm{T}}(k)=\bC_{u_{\rm{T}}}\bx_{u_{\rm{T}}}^{\rm{d}}(k)+e_{\rm o}(k), \\
    \end{dcases}
\end{equation}
where $\bx_{u_{\rm{T}}}(k)\in\mathbb{R}^{2n_{\rm{a}}-1}$ is the state vector.
The matrices $\bA_{u_{\rm{T}}}^{\rm{d}}$ and $\bC_{u_{\rm{T}}}^{\rm{d}}$ of the system (\ref{eq11}) are given as
\begin{equation}\label{eq12}
\bA_{u_{\rm{T}}}^{\rm{d}}=
\begin{pmatrix}
1 & \bzero & \cdots & \bzero \\
\bzero & \bA_{u_{\rm{T}},\rm{1}}^{\rm{d}} & \ddots & \vdots \\
\vdots & \ddots & \ddots & \bzero \\
\bzero & \cdots & \bzero & \bA_{u_{\rm{T}},n_{\rm{a}}-1}^{\rm{d}} \\
\end{pmatrix}
\end{equation}
and
\begin{equation}\label{eq13}
\bC_{u_{\rm{T}}}^{\rm{d}}=
\begin{pmatrix}
1 & \bC_{u_{\rm{T}},\rm{1}}^{\rm{d}} & \cdots & \bC_{u_{\rm{T}},n_{\rm{a}}-1}^{\rm{d}} \\
\end{pmatrix}.
\end{equation}
The individual block entries in these block matrices can be represented as
\begin{equation}\label{eq14}
\bA_{u_{\rm{T}},i}^{\rm{d}}=
\begin{pmatrix}
0 & -1 \\
1 & 2\cos(i2\pi f_{\rm{T}}T_{\rm{s}}) \\
\end{pmatrix}
\end{equation}
and
\begin{equation}\label{eq15}
\bC_{u_{\rm{T}},i}^{\rm{d}}=
\begin{pmatrix}
1 & 0 \\
\end{pmatrix},
\end{equation}
where $T_{\rm{s}}$ denotes the sampling time.

It should be noted that the matrix $\bA_{u_{\rm{T}}}^{\rm{d}}$ is not stable which contains the eigenvalue $1$.

\subsection{Augmentation}
If the input signal $\bu(t)$ of the system (\ref{eq1}) is a band-limited periodic signal, then according to the continuous-time modeling of a periodic signal described in Subsection \ref{ctpm}, we can obtain a state-space model of the signal $\bu(t)$ as

\begin{equation}\label{eq16}
    \begin{dcases}
        \dot\bx_{\bu}(t)=\bA_{\bu}\bx_{\bu}(t), \\
        \bu(t)=\bC_{\bu}\bx_{\bu}(t)+\be_{\rm o}(t), \\
    \end{dcases}
\end{equation}
where
\[\bu(t)=
\begin{pmatrix}
u_1(t) \\
u_2(t) \\
\vdots \\
u_r(t) \\
\end{pmatrix}\in\mathbb{R}^r,
\bx_{\bu}=
\begin{pmatrix}
\bx_{u_{\rm{1}}}(t) \\
\bx_{u_{\rm{2}}}(t) \\
\vdots \\
\bx_{u_r}(t) \\
\end{pmatrix}\in\mathbb{R}^{n_{\bu}},
\bA_{\bu}=
\begin{pmatrix}
\bA_{u_{\rm1}} & \bzero & \cdots & \bzero \\
\bzero & \bA_{u_{\rm2}} & \ddots & \vdots \\
\vdots & \ddots & \ddots & \bzero \\
\bzero & \cdots & \bzero & \bA_{u_r} \\
\end{pmatrix},\]
\[\bC_{\bu}=
\begin{pmatrix}
\bC_{u_{\rm1}} & \bC_{u_{\rm1}} & \cdots & \bC_{u_r} \\
\end{pmatrix},\]
and the individual block entries of $\bA_{\bu}$ and $\bC_{\bu}$ can be denoted as the forms in (\ref{eq7}) and (\ref{eq8}), respectively. $\be_{\rm o}(t)$ could be colored noise or white noise.

By augmenting the state-space model of the system (\ref{eq1}) with the continuous-time periodic signal model (\ref{eq16}), we can have

\begin{equation}\label{eq17}
    \begin{dcases}
        \dot\bx_{\rm a}(t)=\bA_{\rm a}\bx_{\rm a}(t)+\bK_{\rm a}\be_{\rm o}(t), \\
        \by(t)=\bC_{\rm a}\bx_{\rm a}(t)+\bV_{\rm a}\be_{\rm o}(t)+\be_{\rm m}(t), \\
    \end{dcases}
\end{equation}
where
\[\by(t)=
\begin{pmatrix}
y_1(t) \\
y_2(t) \\
\vdots \\
y_m(t) \\
\end{pmatrix},
\bx_{\rm a}(t)=
\begin{pmatrix}
\bx_{\rm s}(t) \\
\bx_{\bu}(t) \\
\end{pmatrix},
\bA_{\rm a}=
\begin{pmatrix}
\bA_{\rm s} & \bB_{\rm s}\bC_{\bu} \\
\bzero & \bA_{\bu} \\
\end{pmatrix},
\bC_{\rm a}=
\begin{pmatrix}
\bC_{\rm s} & \bD_{\rm s}\bC_{\bu} \\
\end{pmatrix},\]
and
\[\bK_{\rm a}=
\begin{pmatrix}
\bB_{\rm s}\\
\bzero
\end{pmatrix},
\bV_{\rm a}=
\bD_{\rm s}.\]

Analogously, the augmented model in discrete time can be expressed as
\begin{equation}\label{eq18}
    \begin{dcases}
        \bx_{\rm a}(k+1)=\bA_{\rm a}^{\rm d}\bx_{\rm a}(k)+\bK_{\rm a}\be_{\rm o}(k), \\
        \by(k)=\bC_{\rm a}^{\rm d}\bx_{\rm a}(k)+\bV_{\rm a}\be_{\rm o}(k)+\be_{\rm m}(k), \\
    \end{dcases}
\end{equation}
where
\[\by(k)=
\begin{pmatrix}
y_1(k) \\
y_2(k) \\
\vdots \\
y_m(k) \\
\end{pmatrix},
\bx_{\rm a}(k)=
\begin{pmatrix}
\bx_{\rm s}(k) \\
\bx_{\bu}(k) \\
\end{pmatrix},
\bA_{\rm a}^{\rm d}=
\begin{pmatrix}
\bA_{\rm s}^{\rm d} & \bB_{\rm s}^{\rm d}\bC_{\bu}^{\rm d} \\
\bzero & \bA_{\bu}^{\rm d} \\
\end{pmatrix},
\bC_{\rm a}^{\rm d}=
\begin{pmatrix}
\bC_{\rm s}^{\rm d} & \bD_{\rm s}^{\rm d}\bC_{\bu}^{\rm d} \\
\end{pmatrix},\]
and
\[\bK_{\rm a}=
\begin{pmatrix}
\bB_{\rm s}^{\rm d}\\
\bzero
\end{pmatrix},
\bV_{\rm a}=
\bD_{\rm s}^{\rm d},\]
and the matrices $\bA_{\rm s}^{\rm d}$, $\bB_{\rm s}^{\rm d}$, $\bC_{\rm s}^{\rm d}$, and $\bD_{\rm s}^{\rm d}$ are discrete-time versions of the matrices $\bA_{\rm s}$, $\bB_{\rm s}$, $\bC_{\rm s}$, and $\bD_{\rm s}$, respectively. The matrices $\bA_{\bu}^{\rm d}$ and $\bC_{\bu}^{\rm d}$ can be formulated by using the discrete-time periodic signal modeling introduced in Subsection \ref{dtpm}. The individual block entries of $\bA_{\bu}$ and $\bC_{\bu}$ can be denoted as the forms in (\ref{eq12}) and (\ref{eq13}), respectively.

\section{Blind identification}\label{s4}
In this section, a novel blind identification method is derived. Specifically, a periodic signal-based system identification (PSSID) method is first proposed, and then based on which the blind identification is proposed.

\subsection{Periodic signal-based system identification}\label{PSSID}

\begin{figure}[ht]
\centering
\includegraphics[scale=0.5,angle=0]{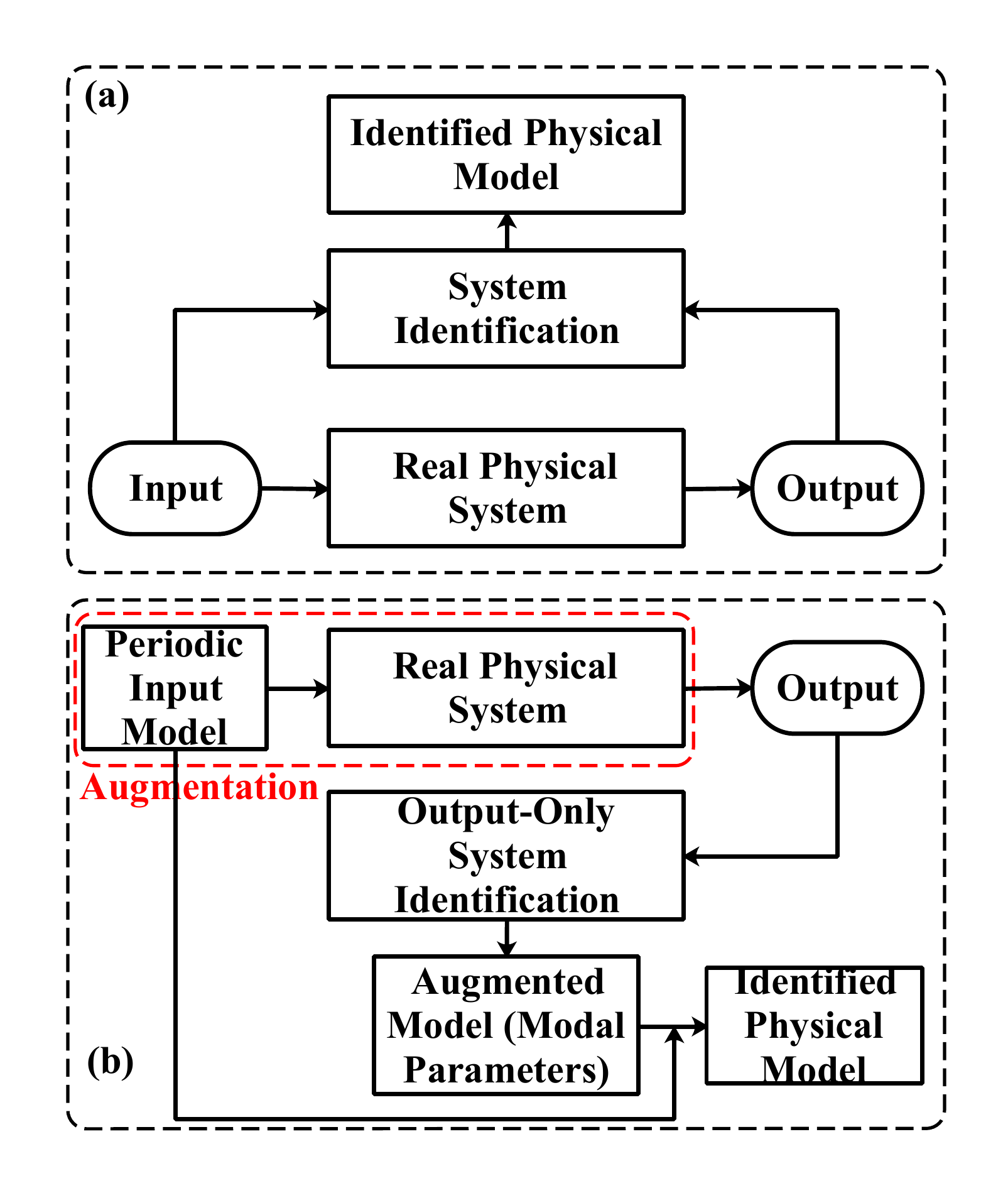}
\caption{Block diagrams for (a) standard system identification and (b) PSSID.}
\label{fig1}
\end{figure}

Different with standard system identification shown in Figure \ref{fig1}, based on using a band-limited periodic signal as the input excitation the proposed PSSID method shown in Figure \ref{fig1} first represents the periodic input as the output of an autonomous state-space model, and then augments the real physical system (\ref{eq1}) with the input signal model, and then uses the purely stochastic subspace identification method (see, e.g., \cite{Overschee1996}) to identify a modal, and finally based on a transformation the matrices $\bA_{\rm{s}}$, $\bC_{\rm{s}}$, $\bD_{\rm{s}}$, and $\bB_{\rm{s}}\bu(t)$ in the system (\ref{eq1}) can be obtained, afterwards using the input data, $\bB_{\rm{s}}$ can be estimated. Moreover, it is worth notice that with knowing only the frequency components of the periodic input, based on the known output signal this novel identification method can extract the eigenvalues of the system (\ref{eq1}), and then according to these eigenvalues the modal parameters including natural frequencies and structural damping ratio of the physical system (\ref{eq1}) can be obtained.

Before illustrating the specific PSSID method, the concept of persistent excitation is first introduced as follows.

\begin{definition}\label{definitiona}
\cite{Verhaegen2007} An ergodic sequence $\bu(k)\in\mathbb{R}^m$ is persistently exciting of order $i$ from time instant $j$ on, if and only if:
\begin{equation*}
\rho(\frac{1}{N}\bU_{j,i,N}\bU_{j,i,N}^{\rm T})=mi, 
\end{equation*}
where the symbol ``\; $^{\rm T}$\; '' denote the matrix transpose, $\rho(.)$ denotes the rank of a square matrix, and
\[\bU_{j,i,N}=
\begin{pmatrix}
\bu(j) & \bu(j+1) & \cdots & \bu(j+N-1) \\
\bu(j+1) & \bu(j+2) & \cdots & \bu(j+N) \\
\vdots & \vdots & \ddots & \vdots \\
\bu(j+i-1) & \bu(j+i) & \cdots & \bu(j+N+i-2) \\
\end{pmatrix}.\]
\end{definition}

The developed PSSID method contains three steps:
\begin{enumerate}[{(1)}]
\item
Use the stochastic subspace identification method to identify the augmented model (\ref{eq18}), and then transform it into its continuous-time version.
\item
Based on the continuous-time model, a real Jordan form of the augmented model is derived, and then eigenvalues of the matrix $\bA_{\rm{a}}$ can be calculated directly, and as a result the modal parameters including natural frequencies and structural damping ratio of the physical system (\ref{eq1}). Afterwards, a canonical form of the identified augmented model can be obtained based on its real Jordan form. Meanwhile, the real Jordan forms of both input signal model (\ref{eq6}) and physical model (\ref{eq1}) can be derived, and then combine them to obtain a theoretical augmented model in canonical form. The definitions of canonical form and Jordan form are introduced in Section \ref{PSSID2}.
\item
By making a comparison between the identified canonical-form model and the theoretical canonical-form model, the physical parameters are uniquely obtained, i.e., a global solution of the system identification can be obtained.
\end{enumerate}

\subsubsection{Step I-Augmented model identification}\label{PSSID1}
For the model (\ref{eq18}), we can transform it into the following form (for how to transform, see the content in Chapter $9.6$ of \cite{Verhaegen2007})

\begin{equation}\label{eq19}
    \begin{dcases}
        \bx_{\rm a}(k+1)=\bA_{\rm a}^{\rm d}\bx_{\rm a}(k)+\bomega(k), \\
        \by(k)=\bC_{\rm a}^{\rm d}\bx_{\rm a}(k)+\bupsilon(k), \\
    \end{dcases}
\end{equation}
where $\bomega(k)$ and $\bupsilon(k)$ denote the process noise and the measurement noise, respectively. $\bomega(k)$ and $\bupsilon(k)$ are zero-mean white noise with covariance matrix
\begin{equation*}
\mathbf{E}(
\begin{pmatrix}
\bomega_p \\
\bupsilon_p \\
\end{pmatrix}
\begin{pmatrix}
\bomega_q^{\rm T} & \bupsilon_q^{\rm T} \\
\end{pmatrix})
=
\begin{pmatrix}
\bQ & \bS \\
\bS^{\rm T} & \bR \\
\end{pmatrix}
\delta_{pq} > \bzero\,,
\end{equation*}
where $\mathbf{E}(\cdot)$ denotes the expected value of a random variable and
$\delta_{pq}=
\begin{dcases}
0, p\neq q, \\
1, p = q. \\
\end{dcases}$

According to the persistent excitation condition in Definition \ref{definitiona}, $\bomega(k)$ (white noise) can excite the model (\ref{eq19}) sufficiently \cite{Verhaegen2007}.

Several assumptions are made for the model (\ref{eq19}).

\begin{assumption}
$\bA_{\rm s}^{\rm d}$ is a stable matrix.
\end{assumption}

\begin{assumption}
The matrix pair $(\bA_{\rm a}^{\rm d}, \bC_{\rm a}^{\rm d})$ is assumed to be observable, which implies that all modes in the system can be observed in the output $\by(k)$ and can thus be identified. The matrix pair $(\bA_{\rm a}^{\rm d}, \bQ^{\frac{1}{2}})$ is assumed to be controllable, which in turn implies that all modes of the system are excited by the stochastic input $\bomega(k)$.
\end{assumption}

\begin{assumption}
The process noise $\bomega(k)$ and the measurement noise $\bupsilon(k)$ are not identically zero.
\end{assumption}

\begin{assumption}
$\bx_{\rm a}(k)$ is independent of $\bomega(k)$ and $\bupsilon(k)$.
\end{assumption}

Since the signal $\bu(t)$ is periodical, and $\bomega(k)$ and $\bupsilon(k)$ are white noise, the stochastic process is stationary, and we can have
\begin{equation*}
    \begin{dcases}
        \mathbf{E}(\bx_{\rm a}(k))=\bC_1, \\
        \mathbf{E}((\bx_{\rm a}(k)-\mathbf{E}(\bx_{\rm a}(k)))(\bx_{\rm a}(k)-\mathbf{E}(\bx_{\rm a}(k)))^{\rm T})=\bC_2, \\
    \end{dcases}
\end{equation*}
where $\bC_1$ and $\bC_2$ are constant matrices.

We transform the model (\ref{eq19}) into the following dynamical system with different coordinates
\begin{equation}\label{eq19a}
    \begin{dcases}
        \bx_{\rm a}^*(k+1)=\bA_{\rm a}^{\rm d}\bx_{\rm a}^*(k)+\bomega(k), \\
        \by(k)=\bC_{\rm a}^{\rm d}\bx_{\rm a}^*(k)+\bupsilon(k)+\bd, \\
    \end{dcases}
\end{equation}
where $\bd$ denotes a constant vector, and $\bd=\bC_1$, and $\bx_{\rm a}^*(k)=\bx_{\rm a}(k)-\bC_1$.

Based on (\ref{eq19}) and (\ref{eq19a}), we can obtain
\begin{align*}
&\mathbf{E}(\bx_{\rm a}^*(k))=\bzero, \\
&\bSigma\triangleq\mathbf{E}(\bx_{\rm a}^*(k)(\bx_{\rm a}^*(k))^{\rm T})=\bA_{\rm a}^{\rm d}\bSigma(\bA_{\rm a}^{\rm d})^{\rm T}+\bQ, \\
&\bLambda\triangleq\mathbf{E}(\by(k)\by^{\rm T}(k))=\bC_{\rm a}^{\rm d}\bSigma(\bC_{\rm a}^{\rm d})^{\rm T}+\bR+\bd\bd^{\rm T}=\bC_{\rm a}^{\rm d}\bSigma(\bC_{\rm a}^{\rm d})^{\rm T}+\bR^*, \\
&\bG\triangleq\mathbf{E}(\bx_{\rm a}^*(k+1)\by^{\rm T}(k))=\bA_{\rm a}^{\rm d}\bSigma(\bC_{\rm a}^{\rm d})^{\rm T}+\bS,
\end{align*}
where $\bR+\bd\bd^{\rm T}=\bR^*$.

Based on the above assumptions and results, a conventional stochastic subspace identification method (see ``stochastic algorithm $3$'' in page $90$ of \cite{Overschee1996} which can always identify a positive real covariance sequence) can be implemented to identify $\bA_{\rm a}^{\rm d}$ and $\bC_{\rm a}^{\rm d}$. (We denote the estimates of $\bA_{\rm a}^{\rm d}$ and $\bC_{\rm a}^{\rm d}$ as $\hat{\bA}_{\rm a}^{\rm d}$ and $\hat{\bC}_{\rm a}^{\rm d}$, respectively. The symbol ``\; $\hat{}$\; " denotes the estimation value.) The complete procedure of the ``stochastic algorithm $3$'' goes through several stages:
\begin{enumerate}[{(1)}]
\item By applying a Kalman filter to the stochastic model (\ref{eq19a}), of which the innovation model can be obtained as
\begin{equation}\label{eq19b}
    \begin{dcases}
        \bx_{\rm a}^{\rm f}(k+1)=\bA_{\rm a}^{\rm d}\bx_{\rm a}^{\rm f}(k)+\bK^{\rm f}\be_{\rm a}^{\rm f}(k), \\
        \by(k)=\bC_{\rm a}^{\rm d}\bx_{\rm a}^{\rm f}(k)+\be_{\rm a}^{\rm f}(k), \\
    \end{dcases}
\end{equation}
where
\begin{equation*}
\bP\triangleq\mathbf{E}(\bx_{\rm a}^{\rm f}(k)(\bx_{\rm a}^{\rm f}(k))^{\rm T})
\end{equation*}
and
\begin{equation*}
\mathbf{E}(\be_{\rm a}^{\rm f}(k)(\be_{\rm a}^{\rm f}(k)))=\bLambda-\bC_{\rm a}^{\rm d}\bP(\bC_{\rm a}^{\rm d})^{\rm T}.
\end{equation*}

Additionally, based on (\ref{eq19a}) and (\ref{eq19b}), we can obtain the following expressions of forward (algebraic) Riccati eqution and steady-state Kalman filter gain $\bK^{\rm f}$ (see page $62$ of \cite{Overschee1996})
\begin{equation}\label{eq19c}
\bP=\bA_{\rm a}^{\rm d}\bP(\bA_{\rm a}^{\rm d})^{\rm T}+(\bG-\bA_{\rm a}^{\rm d}\bP(\bC_{\rm a}^{\rm d})^{\rm T}(\bLambda-\bC_{\rm a}^{\rm d}\bP(\bC_{\rm a}^{\rm d})^{\rm T})^{-1}(\bG-\bA_{\rm a}^{\rm d}\bP(\bC_{\rm a}^{\rm d})^{\rm T})^{\rm T}
\end{equation}
and
\begin{equation}\label{eq19d}
\bK^{\rm f}=(\bG-\bA_{\rm a}^{\rm d}\bP(\bC_{\rm a}^{\rm d})^{\rm T})(\bLambda-\bC_{\rm a}^{\rm d}\bP(\bC_{\rm a}^{\rm d})^{\rm T})^{-1}.
\end{equation}

\item Calculate the projections:
\begin{equation*}
\mathcal{O}_i=\bY_f/\bY_p
\end{equation*}
and
\begin{equation*}
\mathcal{O}_{i-1}=\bY_f^-/\bY_p^+,
\end{equation*}
where
\[
\bY_p
=
\begin{pmatrix}
\by(0) & \by(1) & \cdots & \by(j-1) \\
\vdots & \vdots & \ddots & \vdots \\
\by(i-2) & \by(i-1) & \cdots & \by(i+j-3) \\
\by(i-1) & \by(i) & \cdots & \by(i+j-2)
\end{pmatrix},
\bY_f
=
\begin{pmatrix}
\by(i) & \by(i+1) & \cdots & \by(i+j-1) \\
\vdots & \vdots & \ddots & \vdots \\
\by(2i-2) & \by(2i-1) & \cdots & \by(2i+j-3) \\
\by(2i-1) & \by(2i) & \cdots & \by(2i+j-2)
\end{pmatrix}
\]
and
\begin{equation*}
\bY_p^+
=
\begin{pmatrix}
\by(0) & \by(1) & \cdots & \by(j-1) \\
\vdots & \vdots & \ddots & \vdots \\
\by(i-2) & \by(i-1) & \cdots & \by(i+j-3) \\
\by(i-1) & \by(i) & \cdots & \by(i+j-2) \\
\by(i) & \by(i+1) & \cdots & \by(i+j-1)
\end{pmatrix},
\end{equation*}
\begin{equation*}
\bY_f^-
=
\begin{pmatrix}
\by(i+1) & \by(i+2) & \cdots & \by(i+j) \\
\vdots & \vdots & \ddots & \vdots \\
\by(2i-2) & \by(2i-1) & \cdots & \by(2i+j-3) \\
\by(2i-1) & \by(2i) & \cdots & \by(2i+j-2)
\end{pmatrix}.
\end{equation*}

\item Calculate the SVD of the weighted projection:
\begin{equation*}
\bW_1\mathcal{O}_i\bW_2=\bU\bS\bV^{\rm T},
\end{equation*}
where
\[\bU=
\begin{pmatrix}
  \bU_1 & \bU2
\end{pmatrix},
\bS=
\begin{pmatrix}
  \bS_1 & \bzero \\
  \bzero & \bzero \\
\end{pmatrix},
\bV^{\rm T}=
\begin{pmatrix}
  \bV_1^{\rm T} \\
  \bV_2^{\rm T}
\end{pmatrix}.\]

Based on the matrix $\bS_1$, we can determine the model order of (\ref{eq19a}).

\begin{remark}
Because the order of the model (\ref{eq16}) is known, we only use the SVD technique to determine the order of the system (\ref{eq19a}).
\end{remark}

\item Calculate the extended observability matrices:\\
The transformation from the discrete-time system (\ref{eq19a}) to its continuous-time version can be done as follows
\begin{equation}\label{eq20}
\bA_{\rm a}=\frac{1}{T_{\rm s}}\ln(\bA_{\rm a}^{\rm d})
\end{equation}
and
\begin{equation}\label{eq21}
\bC_{\rm a}=\bC_{\rm a}^{\rm d}.
\end{equation}

Here it should be noted that the zero-order hold method cannot handle discrete-time systems with poles which are zero, in addition, the zero-order hold increases the model order for discrete systems with negative real poles \cite{Franklin1997}.

With the identified matrices $\hat{\bA}_{\rm a}^{\rm d}$ and $\hat{\bC}_{\rm a}^{\rm d}$, and Equations (\ref{eq20}) and (\ref{eq21}), we can obtain $\hat{\bA}_{\rm a}$ and $\hat{\bC}_{\rm a}$. The reason why we need the continuous-time version of the model (\ref{eq19a}) is that a physical model needs to be derived. For the calculation of the initial state value $\bx_{\rm a}^*(0)$ or the estimation of the state $\bx_{\rm a}^*(k)$ ($k>0$), we can use the Kalman filter for the estimated augmented model and the output signal $\by(k)$. Since the augmented model (\ref{eq19a}) without noise is a linear output-only model, the system transition and measurement sensitivity matrices $(\bA_{\rm a}, \bC_{\rm a})$ can be used to represent the whole augmented model (\ref{eq19a}) without noise effect, and therefore hereafter, we refer to different forms of linear output-only model based on these two matrices.

\begin{remark}
If there is a periodic input signal with offset used for the identification, the matrix $\bA_{\bu}^{\rm d}$ will be an unstable matrix. This is obvious since the offset corresponds to the number ``1" in the matrix $\bA_{\bu}^{\rm d}$.
\end{remark}

\begin{remark}
Because of the augmentation, we need to be careful about the observability of the model (\ref{eq19}), the relationship between the observability of the system (\ref{eq19}) and the frequency components of the signal model can be found in \ref{appendixb}. The relationship between the input number and output number is discussed therein.
\end{remark}
\end{enumerate}

\subsubsection{Step II-Real Jordan form and canonical form}\label{PSSID2}
In this section, one definition related to canonical form is first given as follows.
\begin{definition}\label{definitionb}
\cite{Wang1976} Let $X$ be a set, and let $E$ be an equivalence relation on $X$. A map $\phi:X \rightarrow X$ is said to be a canonical map for $E$ on $X$ if \\
i) $xE\phi(x),\forall\,x \in X$, \\
ii) $xEy\Leftrightarrow\phi(x)=\phi(y),\forall\,x,y\in X$. \\
The image of $\phi$, denoted by $\mathrm{Im}\,\phi$, is said to be a set of canonical forms for $E$ on $X$. It is also said to be a set of $E$-canonical forms on $X$.
\end{definition}

According to the above definition, we can conclude that a square matrix can have several kinds of canonical forms, such as real Jordan (canonical) form and the canonical form obtained by transforming the real Jordan form. In this paper, we call the latter one canonical form to make a difference (even though the Jordan form is also a kind of canonical form), and in the following this form is introduced once more. Furthermore, the specific description of the real Jordan form can be found in \ref{appendixa}.

The system transition matrices of both input signal model (\ref{eq16}) and physical model (\ref{eq3}) can be respectively decomposed into their real Jordan forms as follows
\begin{equation}\label{eq22}
\bA_{\bu}=\bT_{\bu}\bA_{\bu}^{\rm J}\bT_{\bu}^{-1},
\end{equation}
and
\begin{equation}\label{eq23}
\bA_{\rm s}=\bT_{\rm s}\bA_{\rm s}^{\rm J}\bT_{\rm s}^{-1},
\end{equation}
where $\bT_{\bu}$ and $\bT_{\rm s}$ are transformation matrices, and according to the model structure of the signal model (\ref{eq16}) the matrix $\bT_{\bu}$ is a identity matrix, i.e., $\bA_{\bu}=\bA_{\bu}^{\rm J}$.

Based on (\ref{eq22}) and (\ref{eq23}), both input signal model (\ref{eq16}) and system (\ref{eq3}) without noise effect can be realized in their real Jordan forms as follows
\begin{equation}\label{eq24}
    \begin{dcases}
        \dot\beeta_{\bu}(t)=\bA_{\bu}^{\rm J}\beeta_{\bu}(t), \\
        \bu(t)=\bC_{\bu}^{\rm J}\beeta_{\bu}(t), \\
    \end{dcases}
\end{equation}
and

\begin{equation}\label{eq25}
    \begin{dcases}
        \dot\beeta_{\rm s}(t)=\bA_{\rm s}^{\rm J}\beeta_{\rm s}(t)+\bB_{\rm s}^{\rm J}\bu(t), \\
        \by(t)=\bC_{\rm s}^{\rm J}\beeta_{\rm s}(t)+\bD_{\rm s}^{\rm J}\bu(t), \\
    \end{dcases}
\end{equation}
where $\bx_{\bu}(t)=\bT_{\bu}\beeta_{\bu}(t)$, $\bx_{\rm s}(t)=\bT_{\rm s}\beeta_{\rm s}(t)$, $\bC_{\bu}^{\rm J}=\bC_{\bu}\bT_{\bu}$, $\bB_{\rm s}^{\rm J}=\bT_{\rm s}^{-1}\bB_{\rm s}$, $\bC_{\rm s}^{\rm J}=\bC_{\rm s}\bT_{\rm s}$, and $\bD_{\rm s}^{\rm J}=\bD_{\rm s}$.

It should be noted that in the models (\ref{eq24}) and (\ref{eq25}), we delete the noise terms in order to reduce the complexity for expressing the formulas, and the deletion does not affect the transformations of the system matrices. We still use ``$=$'' instead of ``$\simeq$'' in (\ref{eq24}) and (\ref{eq25}) even though the approximation notation ``$\simeq$'' is the right one.

Based on (\ref{eq24}) and (\ref{eq25}), an augmented model in theory (We call this model theoretical model.) can be obtained as follows
\begin{equation}\label{eq26}
    \begin{dcases}
        \dot\beeta_{\rm a}(t)=\bA_{\rm a}^{\rm J}\beeta_{\rm a}(t), \\
        \by(t)=\bC_{\rm a}^{\rm J}\beeta_{\rm a}(t), \\
    \end{dcases}
\end{equation}
where $\beeta_{\rm a}(t)=
\begin{pmatrix}
\beeta_{\rm s}(t) \\
\beeta_{\bu}(t) \\
\end{pmatrix}$, $\bA_{\rm a}^{\rm J}=
\begin{pmatrix}
\bA_{\rm s}^{\rm J} & \bB_{\rm s}^{\rm J}\bC_{\bu}^{\rm J} \\
\bzero & \bA_{\bu}^{\rm J} \\
\end{pmatrix}$, and
$\bC_{\rm a}^{\rm J}=
\begin{pmatrix}
\bC_{\rm s}^{\rm J} & \bD_{\rm s}^{\rm J}\bC_{\bu}^{\rm J} \\
\end{pmatrix}$.

Equations (\ref{eq24}) and (\ref{eq25}) represent real Jordan realizations for the input signal model (\ref{eq16}) and the system (\ref{eq3}) in theory, and by combing these two real Jordan realization, we can obtain an augmented model (\ref{eq26}) in another canonical form, as aforementioned, this canonical form can be obtained by transforming the Jordan form of the augmented model. As mentioned in the start of Section \ref{PSSID}, the theoretical augmented model (\ref{eq26}) is needed to make comparison with the identified augmented model in canonical form, in the following, the derivation of the canonical form of the identified augmented model is illustrated, and the comparison is made in Step III.

After the identification of the system (\ref{eq19a}), we can transform the identified matrix $\hat\bA_{\rm a}$ into its real Jordan form $\hat\bA_{\rm J}$ using a transformation matrix $\bT_{\rm a}$, that is
\begin{equation}\label{eq27}
\hat\bA_{\rm a}=\bT_{\rm a}\hat\bA_{\rm J}\bT_{\rm a}^{-1}.
\end{equation}
The matrix $\hat\bA_{\rm J}$ is a block diagonal matrix which can de expressed as follows
\begin{equation}
\hat\bA_{\rm J}=
    \begin{pmatrix}\label{eq28}
        \bJ_{\bA}^{(1)} & \bzero & \cdots & \bzero \\
        \bzero & \bJ_{\bA}^{(2)} & \cdots & \bzero \\
        \vdots & \vdots & \ddots & \vdots \\
        \bzero & \bzero & \cdots & \bJ_{\bA}^{(N)} \\
    \end{pmatrix},
\end{equation}
where $\bJ_{\bA}^{(k)}$ corresponds to a single eigenfrequency.

Different real Jordan decompositions of $\hat\bA_{\rm a}$ are the order variations of the Jordan blocks ($\bJ_{\bA}^{(k)},k=1, 2, \ldots, N$) in the matrix $\hat\bA_{\rm J}$. Since the periodic input excitation frequency components are known prior, it can be easily to make a difference between the identified frequency components of the input and the identified frequency components of the system. Sometimes, if a case that some identified frequency components of the input are the same as or nearly the same as some identified frequency components of the system is encountered, we can change the frequency components of the input excitation. Hence it is unnecessary to make an assumption that there are no frequency components overlaps between the signal model and the physical model.

Meanwhile, the matrix $\hat\bC_{\rm a}$ and the estimated initial state $\hat\bx_{\rm a}^*(0)$ should also be transformed as follows
\begin{equation}\label{eq29}
\hat\bC_{\rm J}=\hat\bC_{\rm a}\bT_{\rm a},
\end{equation}
and
\begin{equation}\label{eq30}
\hat\bx_{\rm J}(0)=\bT_{\rm a}^{-1}\hat\bx_{\rm a}^*(0),
\end{equation}
where $\hat\bC_{\rm J}$ and $\hat\bx_{\rm a}^*(0)$ together with $\hat\bA_{\rm J}$ correspond to the real Jordan form of the identified model of the system (\ref{eq19a}).

Furthermore, we can also make a separation of the frequency components of the input signal model (\ref{eq16}) and the physical model (\ref{eq3}) since the frequency components of the input model are known, as a result by using a transformation matrix we can obtain a new form of the matrix $\hat\bA_{\rm J}$ as follow
\begin{equation}\label{eq31}
\bar\bA_{\rm a}^{\rm J}=\bT_{\rm J}\hat\bA_{\rm J}\bT_{\rm J}^{-1}=
\begin{pmatrix}
\hat\bJ_{\rm s} & \bzero \\
\bzero & \hat\bJ_{\bu} \\
\end{pmatrix},
\end{equation}
where $\hat\bJ_{\rm s}$ contains the Jordan blocks corresponding to the physical system while $\hat\bJ_{\bu}$ contains the Jordan blocks corresponding to the input signal, and $\hat\bJ_{\rm s}$ and $\hat\bJ_{\bu}$ can be seen the identified values of $\bA_{\rm s}^{\rm J}$ and $\bA_{\bu}^{\rm J}$. The matrix $\bT_{\rm J}$ denotes the transformation matrix.

Accordingly, the transformation for $\hat\bC_{\rm a}^{\rm J}$ and the initial estimated state $\hat\bx_{\rm a}^{\rm J}(0)$ are
\begin{equation}\label{eq32}
\bar\bC_{\rm a}^{\rm J}=\hat\bC_{\rm J}\bT_{\rm J}^{-1}=
\begin{pmatrix}
\hat\bC_{\rm s}^{\rm J} & \hat\bC_{\bu}^{\rm J} \\
\end{pmatrix},
\end{equation}
and
\begin{equation}\label{eq33}
\bar\bx_{\rm a}^{\rm J}(0)=\bT_{\rm J}\hat\bx_{\rm J}(0)=
\begin{pmatrix}
\bar\bx_{\rm s}^{\rm J}(0) \\
\bar\bx_{\bu}^{\rm J}(0) \\
\end{pmatrix}.
\end{equation}

Based on the above analysis, the estimated modal parameters including natural frequencies, natural modes and structural damping ratios of the physical system can be obtained directly from $\hat\bJ_{\rm s}$. Denote the $i$-th natural frequency of the physical system (\ref{eq1}) as $f_{{\rm nat},i}$, denote the $i$-th natural mode of the physical system (\ref{eq1}) as $\bq_i$, and denote the $i$-th damping ratio of the physical system $(\ref{eq1})$ as $\zeta_i$, the natural modes can be calculated directly using the eigenvalues, and according to \cite{Shi2016}, $f_{{\rm nat},i}$ and $\zeta_i$ can be calculated using the following formulas
\begin{equation}\label{eq34}
f_{{\rm nat},i}=\sqrt{(\sigma_i^{\rm s})^2+(\omega_i^{\rm s})^2},
\end{equation}
and
\begin{equation}\label{eq35}
\zeta_i=\frac{\left|\sigma_i^{\rm s}\right|}{\sqrt{(\sigma_i^{\rm s})^2+(\omega_i^{\rm s})^2}},
\end{equation}
where $|\cdot|$ denotes the absolute value, and the complex values $\sigma_i^{\rm s}\pm j\omega_i^{\rm s}$ denotes the eigenvalues of the matrix $\bA_{\rm s}^{\rm J}$ in a conjugate pair.

Based on the above results, using Roth's removal rule \cite{Gerrish1998}, we can transform the system realization $(\bar\bA_{\rm a}^{\rm J},\bar\bC_{\rm a}^{\rm J})$ into the realization $(\hat\bA_{\rm a}^{\rm J},\hat\bC_{\rm a}^{\rm J})$ using the transformation matrix $\bT=\begin{pmatrix}\bI & \bX \\ \bzero & \bI \end{pmatrix}$, that is
\begin{equation}\label{eq36}
\hat\bA_{\rm a}^{\rm J}=\bT\bar\bA_{\rm a}^{\rm J}\bT^{-1}=
\begin{pmatrix}
\hat\bJ_{\rm s} & -\hat\bJ_{\rm s}\bX+\bX\hat\bJ_{\bu} \\
\bzero & \hat\bJ_{\bu}
\end{pmatrix}
\end{equation}
and
\begin{equation}\label{eq37}
\hat\bC_{\rm a}^{\rm J}=\bar\bC_{\rm a}^{\rm J}\bT^{-1}=
\begin{pmatrix}
\bar\bC_{\rm s}^{\rm J} & -\bar\bC_{\rm s}^{\rm J}\bX+\bar\bC_{\bu}^{\rm J}
\end{pmatrix},
\end{equation}
where the matrices $\hat\bA_{\rm a}^{\rm J}$ and $\hat\bC_{\rm a}^{\rm J}$ can be respectively seen as the identified matrices $\bA_{\rm a}^{\rm J}$ and $\bC_{\rm a}^{\rm J}$. The uniqueness of the matrix $\bX$ is discussed in the next section.

Accordingly, the estimated initial state can be transformed as follows
\begin{equation}\label{eq38}
\hat\bx_{\rm a}^{\rm J}(0)=\bT\bar\bx_{\rm a}^{\rm J}(0)=
\begin{pmatrix}
\bar\bx_{\rm s}^{\rm J}(0)+\bX\bar\bx_{\bu}^{\rm J}(0) \\
\bar\bx_{\bu}^{\rm J}(0)
\end{pmatrix}.
\end{equation}

\subsubsection{Step III-Identification of the state-space model in physical coordinates}\label{PSSID3}
Based on the theoretical canonical-form augmented model and the identified one obtained in Step II, we can compare these two model, i.e., we let the identified one be equal to the theoretical one, then we can obtain the following equations
\begin{equation}\label{eq39}
\bA_{\rm s}\bT_{\rm s}=\bT_{\rm s}\hat\bJ_{\rm s},
\end{equation}

\begin{equation}\label{eq40}
\bC_{\rm s}\bT_{\rm s}=\hat\bC_{\rm s}^{\rm J},
\end{equation}

\begin{equation}\label{eq41}
\bT_{\rm s}(-\hat\bJ_{\rm s}\bX+\bX\hat\bJ_{\bu})=\bB_{\rm s}\bC_{\bu}\bT_{\bu},
\end{equation}

\begin{equation}\label{eq42}
-\bar\bC_{\rm s}^{\rm J}\bX+\bar\bC_{\bu}^{\rm J}=\bD_{\rm s}\bC_{\bu}\bT_{\bu}.
\end{equation}

Representing $\bT_{\rm s}$, $\hat\bJ_{\rm s}$, and $\bar\bC_{\rm s}^{\rm J}$ in the partitioned forms as follows
\begin{equation}\label{eq43}
\bT_{\rm s}=
\begin{pmatrix}
(\bT_{\rm s})_{11} & (\bT_{\rm s})_{12} \\
(\bT_{\rm s})_{21} & (\bT_{\rm s})_{22} \\
\end{pmatrix},
\end{equation}

\begin{equation}\label{eq44}
\hat\bJ_{\rm s}=
\begin{pmatrix}
(\hat\bJ_{\rm s})_{11} & (\hat\bJ_{\rm s})_{12} \\
(\hat\bJ_{\rm s})_{21} & (\hat\bJ_{\rm s})_{22} \\
\end{pmatrix},
\end{equation}
and
\begin{equation}\label{eq45}
\hat\bC_{\rm s}=
\begin{pmatrix}
(\hat\bC_{\rm s})_{11} & (\hat\bC_{\rm s})_{12} \\
\end{pmatrix}.
\end{equation}

Based on (\ref{eq39}), (\ref{eq40}), (\ref{eq41}), (\ref{eq42}), and (\ref{eq43}), we can obtain
\begin{equation}\label{eq46}
\begin{pmatrix}
\bzero & \bI \\
-\bM_{-1}\bK & -\bM_{-1}\bD \\
\end{pmatrix}
\begin{pmatrix}
(\bT_{\rm s})_{11} & (\bT_{\rm s})_{12} \\
(\bT_{\rm s})_{21} & (\bT_{\rm s})_{22} \\
\end{pmatrix}
=
\begin{pmatrix}
(\bT_{\rm s})_{11} & (\bT_{\rm s})_{12} \\
(\bT_{\rm s})_{21} & (\bT_{\rm s})_{22} \\
\end{pmatrix}
\begin{pmatrix}
(\hat\bJ_{\rm s})_{11} & (\hat\bJ_{\rm s})_{12} \\
(\hat\bJ_{\rm s})_{21} & (\hat\bJ_{\rm s})_{22} \\
\end{pmatrix}
\end{equation}
and
\begin{equation}\label{eq47}
\begin{pmatrix}
\bC_{\rm p}-\bC_{\rm ac}\bM^{-1}\bK & \bC_{\rm v}-\bC_{\rm ac}\bM^{-1}\bD
\end{pmatrix}
\begin{pmatrix}
(\bT_{\rm s})_{11} & (\bT_{\rm s})_{12} \\
(\bT_{\rm s})_{21} & (\bT_{\rm s})_{22} \\
\end{pmatrix}
=
\begin{pmatrix}
(\hat\bC_{\rm s})_{11} & (\hat\bC_{\rm s})_{12}
\end{pmatrix}.
\end{equation}

Expanding (\ref{eq46}), we can obtain
\begin{equation}\label{eq48}
\begin{dcases}
(\bT_{\rm s})_{21}=(\bT_{\rm s})_{11}(\hat\bJ_{\rm s})_{11}+(\bT_{\rm s})_{12}(\hat\bJ_{\rm s})_{21}, \\
(\bT_{\rm s})_{22}=(\bT_{\rm s})_{11}(\hat\bJ_{\rm s})_{12}+(\bT_{\rm s})_{12}(\hat\bJ_{\rm s})_{22}, \\
\end{dcases}
\end{equation}
and
\begin{equation}\label{eq49}
\begin{dcases}
-\bM^{-1}\bK(\bT_{\rm s})_{11}-\bM^{-1}\bD(\bT_{\rm s})_{21}=(\bT_{\rm s})_{21}(\hat\bJ_{\rm s})_{11}+(\bT_{\rm s})_{22}(\hat\bJ_{\rm s})_{21}, \\
-\bM^{-1}\bK(\bT_{\rm s})_{12}-\bM^{-1}\bD(\bT_{\rm s})_{22}=(\bT_{\rm s})_{21}(\hat\bJ_{\rm s})_{12}+(\bT_{\rm s})_{22}(\hat\bJ_{\rm s})_{22}. \\
\end{dcases}
\end{equation}

Also, expanding (\ref{eq47}), we can obtain
\begin{equation}\label{eq50}
\begin{dcases}
\bC_{\rm p}(\bT_{\rm s})_{11}+\bC_{\rm v}(\bT_{\rm s})_{21}+\bC_{\rm ac}(-\bM^{-1}\bK(\bT_{\rm s})_{11}-\bM^{-1}\bD(\bT_{\rm s})_{21})=(\bar\bC_{\rm s}^{\rm J})_{11}, \\
\bC_{\rm p}(\bT_{\rm s})_{12}+\bC_{\rm v}(\bT_{\rm s})_{22}+\bC_{\rm ac}((\hat\bJ_{\rm s})_{12}+(\bT_{\rm s})_{12}(\hat\bJ_{\rm s})_{22})=(\bar\bC_{\rm s}^{\rm J})_{12}. \\
\end{dcases}
\end{equation}

Based on (\ref{eq49}) and (\ref{eq50}), we can obtain
\begin{equation}\label{eq51}
\begin{dcases}
\bC_{\rm p}(\bT_{\rm s})_{11}+\bC_{\rm v}(\bT_{\rm s})_{21}+\bC_{\rm ac}((\bT_{\rm s})_{11}(\hat\bJ_{\rm s})_{11}+(\bT_{\rm s})_{12}(\hat\bJ_{\rm s})_{21})=(\bar\bC_{\rm s}^{\rm J})_{11}, \\
\bC_{\rm p}(\bT_{\rm s})_{12}+\bC_{\rm v}(\bT_{\rm s})_{22}+\bC_{\rm ac}((\bT_{\rm s})_{21}(\hat\bJ_{\rm s})_{12}+(\bT_{\rm s})_{22}(\hat\bJ_{\rm s})_{22})=(\bar\bC_{\rm s}^{\rm J})_{12}. \\
\end{dcases}
\end{equation}

Considering (\ref{eq48})and (\ref{eq51}) together, enough equations are provided to solve for $\bT_{\rm s}$. The solution $\hat\bT_{\rm s}$ can be obtained as follows via the vectorization rule
\begin{equation}\label{eq52}
\hat\bT_{\rm s}=
\begin{pmatrix}
{\rm vec}((\hat\bT_{\rm s})_{11}) \\
{\rm vec}((\hat\bT_{\rm s})_{12}) \\
{\rm vec}((\hat\bT_{\rm s})_{21}) \\
{\rm vec}((\hat\bT_{\rm s})_{22}) \\
\end{pmatrix}=
\bV^\dag\bW,
\end{equation}
where
$$
\begin{tiny}
\bV=
\begin{pmatrix}
(\hat\bT_{\rm s})_{11}^{\mathrm T}\otimes\bI_{n\times n} & (\hat\bT_{\rm s})_{21}^{\rm T}\otimes\bI_{n\times n}
& -\bI_{n^2\times n^2} & \bzero_{n^2\times n^2}\vspace{1ex} \\
(\hat\bT_{\rm s})_{12}^{\mathrm T}\otimes\bI_{n\times n} & (\hat\bT_{\rm s})_{22}^{\rm T}\otimes\bI_{n\times n}
& \bzero_{n^2\times n^2} & -\bI_{n^2\times n^2}\vspace{1ex} \\
\bI_{n\times n}\otimes\bC_{\rm p} & \bzero_{n^2\times n^2}
& \bI_{n\times n}\otimes\bC_{\rm v}+(\hat\bT_{\rm s})_{11}^{\mathrm T}\otimes\bC_{\rm ac} &
(\hat\bT_{\rm s})_{21}^{\mathrm T}\otimes\bC_{\rm ac}\vspace{1ex} \\
\bzero_{n^2\times n^2} & \bI_{n\times n}\otimes\bC_{\rm p} & (\hat\bT_{\rm s})_{12}\otimes\bC_{\rm ac} &
\bI_{n\times n}\otimes\bC_{\rm v}+(\hat\bT_{\rm s})_{22}^{\mathrm T}\otimes\bC_{\rm ac}
\end{pmatrix},
\end{tiny}$$
and
$$\bW=
\begin{pmatrix}
\bzero_{n^2\times1} \\
\bzero_{n^2\times1} \\
{\rm vec}((\bar\bC_{\rm s}^{\rm J})_{11}) \\
{\rm vec}((\bar\bC_{\rm s}^{\rm J})_{12}) \\
\end{pmatrix}.$$
$(\cdot)^\dag$ denotes the Moore-Penrose pseudo-inverse a matrix. ${\rm vec}(\cdot)$ is the vectorization operator that stacks all columns of a matrix from left to right into a single column from top to down. To obtain the solution of (\ref{eq52}), the rule of vectorization, ${\rm vec}(\bA_1\bA_2\bA_3)=(\bA_3^{\rm T}\otimes\bA_1){\rm vec}(\bA_2)$, is applied on (\ref{eq48}) and (\ref{eq51}).

As shown in (\ref{eq52}), for one unique solution $\hat\bT_{\rm s}$, the matrix $\bV$ should be full column rank, and if the matrix $\bV$ is not full column rank, we can still obtain an unique solution $\hat\bT_{\rm s}$ by changing the type or the number of measurement sensors (e.g., the displacement sensor, the velocity sensor, and the acceleration sensor). As known, the acceleration sensor is cheap and widely used in structural engineering, in the following, an example is given based on using only acceleration sensor to guarantee the uniqueness of the solution $\hat\bT_{\rm s}$.

\begin{exmp}
Setting $\bC_{\rm p}=\bzero_{n\times n}$, $\bC_{\rm v}=\bzero_{n\times n}$, and $\bC_{\rm ac}=\bI_{n\times n}$ yields a new matrix $\bV$ denoted as $\bV_1$, that is

\begin{equation}\label{eq53}
\bV_1=
\begin{pmatrix}
(\hat\bT_{\rm s})_{11}^{\mathrm T}\otimes\bI_{n\times n} & (\hat\bT_{\rm s})_{21}^{\rm T}\otimes\bI_{n\times n}
& -\bI_{n^2\times n^2} & \bzero_{n^2\times n^2}\vspace{1ex} \\
(\hat\bT_{\rm s})_{12}^{\mathrm T}\otimes\bI_{n\times n} & (\hat\bT_{\rm s})_{22}^{\rm T}\otimes\bI_{n\times n}
& \bzero_{n^2\times n^2} & -\bI_{n^2\times n^2}\vspace{1ex} \\
\bzero_{n^2\times n^2} & \bzero_{n^2\times n^2}
& (\hat\bT_{\rm s})_{11}^{\mathrm T}\otimes\bC_{\rm ac} &
(\hat\bT_{\rm s})_{21}^{\mathrm T}\otimes\bC_{\rm ac}\vspace{1ex} \\
\bzero_{n^2\times n^2} & \bzero_{n^2\times n^2} & (\hat\bT_{\rm s})_{12}\otimes\bC_{\rm ac} &
(\hat\bT_{\rm s})_{22}^{\mathrm T}\otimes\bC_{\rm ac}
\end{pmatrix},
\end{equation}
Since 
${\mathrm{det}}(\bV_1)\neq 0$, where $\mathrm{det}(\cdot)$ denotes the determinant of a square matrix, the rank of the matrix $\bV_1$ is $4n^2$, and as a result the left inverse of the matrix $\bV_1$ exists, i.e., the solution $\hat\bT_{\rm s}$ can be uniquely obtained. So in practical application when the matrix $\bV$ is not full column rank, the setting in this example can be used.
\end{exmp}

In this paper, as assumed in Assumption \ref{assum2}, we use the setting in this example, i.e., $\bC_{\rm p}=\bzero_{n\times n}$, $\bC_{\rm v}=\bzero_{n\times n}$, and $\bC_{\rm ac}=\bI_{n\times n}$.

Substituting $\bB_{\rm s}=\begin{pmatrix}\bzero \\ \bM^{-1}\bB\end{pmatrix}$ and $\bD_{\rm s}=\bC_{\rm ca}\bM^{-1}\bB$ into (\ref{eq41}) and (\ref{eq42}) results in
\begin{equation}\label{eq54}
\bT_{\rm s}(-\hat\bJ_{\rm s}\bX+\bX\hat\bJ_{\bu})=
\begin{pmatrix}\bzero \\ \bM^{-1}\bB\end{pmatrix}\bC_\bu\bT_\bu
\end{equation}
and
\begin{equation}\label{eq55}
-\bar\bC_{\rm s}^{\rm J}\bX+\bar\bC_{\bu}^{\rm J}=\bC_{\rm ca}\bM^{-1}\bB\bC_\bu\bT_\bu.
\end{equation}

Based on (\ref{eq54}), we can obtain
\begin{equation}\label{eq56}
\bH_1\bT_{\rm s}(-\hat\bJ_{\rm s}\bX+\bX\hat\bJ_{\bu})=\bzero
\end{equation}
and
\begin{equation}\label{eq57}
\bH_2\bT_{\rm s}(-\hat\bJ_{\rm s}\bX+\bX\hat\bJ_{\bu})=\bM^{-1}\bB\bC_\bu\bT_\bu,
\end{equation}
where $\bH_1=\begin{pmatrix}\bI_{n\times n} & \bzero_{n\times n}\end{pmatrix}$ and $\bH_2=\begin{pmatrix}\bzero_{n\times n} & \bI_{n\times n}\end{pmatrix}$.

Substituting $\bM^{-1}\bB\bC_\bu\bT_\bu$ in the right side of (\ref{eq57}) into (\ref{eq55}) yields
\begin{equation}\label{eq58}
\bC_{\rm ca}\bH_2\bT_{\rm s}(-\hat\bJ_{\rm s}\bX+\bX\hat\bJ_{\bu})=-\bar\bC_{\rm s}^{\rm J}\bX+\bar\bC_{\bu}^{\rm J}.
\end{equation}

Enough equations are provided by (\ref{eq56}) and (\ref{eq58}) to solve for the matrix $\bX$. The solution $\hat\bX$ can be obtained as follows via the vectorization rule
\begin{equation}\label{eq59}
{\rm vec}(\hat\bX)=\bF^\dag\bG,
\end{equation}
where \[\bF=\begin{pmatrix} -\bI_{n_\bu\times n_\bu}\otimes(\bH_1\hat\bT_{\rm s}\hat\bJ_{\rm s})+\hat\bJ_{\bu}^{\rm T}\otimes(\bH_1\hat\bT_{\rm s}) \\
-\bI_{n_\bu\times n_\bu}\otimes(-\bC_{\rm ac}\bH_2\hat\bT_{\rm s}\hat\bJ_{\rm s}+\bar\bC_{\rm s}^{\rm J})+\hat\bJ_{\bu}^{\rm T}\otimes(\bC_{\rm ac}\bH_2\hat\bT_{\rm s})\end{pmatrix}\] and
\[\bG=\begin{pmatrix}\bzero_{nr\times 1} \\ {\rm vec}(\bar\bC_{\bu}^{\rm J})\end{pmatrix}.\]

In practical application, sometimes the matrix $\bF$ may not have a left inverse, but we can make it to have a left inverse through changing the frequency components of the input signal, i.e., by changing the value of $\hat\bJ_{\bu}^{\rm T}$ we can make the matrix $\bF$ to be full column rank, as a result the solution $\hat\bX$ can be guaranteed uniquely.

After calculating the values of $\hat\bT_{\rm s}$ and $\hat\bX$, we can obtain the estimates of $\bA_{\rm s}$ and $\bC_{\rm s}$ as follows
\begin{equation}\label{eq60}
\hat\bA_{\rm s}=\hat\bT_{\rm s}\hat\bJ_{\rm s}\hat\bT_{\rm s}^{-1}
\end{equation}
and
\begin{equation}\label{eq61}
\hat\bC_{\rm s}=\bar\bC_{\rm s}^{\rm J}\hat\bT_{\rm s}^{-1}.
\end{equation}
As for the estimates of $\bB_{\rm s}$ and $\bD_{\rm s}$, we need to do the following steps:
\begin{enumerate}[{(1)}]
\item
Obtain the initial state estimate $\hat\bx_{\rm a}^{\rm J}(0)$ in (\ref{eq38}). 
\item
Based on the discrete-time autonomous system $(\hat\bA_{\rm a}^{\rm d},\hat\bC_{\rm a}^{\rm d})$, we can compute the history of the state estimation $\left\{\hat\bx_{\rm a}^*(k),0\leqslant k\leqslant N\right\}$, and then transform it into $\left\{\bar\bx_{\rm a}^{\rm J}(k),0\leqslant k\leqslant N\right\}$ based on using (\ref{eq30}) and (\ref{eq33}). However, it should be noted that since the Kalman filter for the system $(\hat\bA_{\rm a}^{\rm d},\hat\bC_{\rm a}^{\rm d})$ is involved to estimate the state sequence $\left\{\bx_{\rm a}^*(k),0\leqslant k\leqslant N\right\}$, the forepart of the estimation of this sequence can be inaccurate. So in practical application, we can choose some value not from the forepart as the initial state value for the Kalman filter.
\item
Extract the state estimation sequence $\left\{\bar\bx_{\bu}^{\rm J}(k),0\leqslant k\leqslant N\right\}$ of the input signal model from the sequence $\left\{\bar\bx_{\rm a}^{\rm J}(k),0\leqslant k\leqslant N\right\}$.
\item
Define the effective input signal estimate as $\hat\bf_{\rm e}(k)$, and $\hat\bf_{\rm e}(k)=\hat\bT_{\rm s}(-\hat\bJ_{\rm s}\hat\bX+\hat\bX\hat\bJ_{\bu})\bar\bx_{\bu}^{\rm J}(k)$, then based on $\left\{\bar\bx_{\bu}^{\rm J}(k),0\leqslant k\leqslant N\right\}$, we can compute $\left\{\hat\bf_{\rm e}(k),0\leqslant k\leqslant N\right\}$. \\
\item
According to (\ref{eq41}), we can obtain \\
\begin{equation*}\hat\bT_{\rm s}(-\hat\bJ_{\rm s}\hat\bX+\hat\bX\hat\bJ_{\bu})=\bB_{\rm s}\bC_\bu\bT_\bu,\end{equation*}
and then substituting it into the expression of $\hat\bf_{\rm e}(k)$ yields \\
\begin{equation*}\hat\bf_{\rm e}(k)=\bB_{\rm s}\bC_\bu\bT_\bu\bar\bx_{\bu}^{\rm J}(k),\end{equation*}
which means $\begin{pmatrix} \hat\bf_{\rm e}(0) & \hat\bf_{\rm e}(1) & \cdots & \hat\bf_{\rm e}(N)\end{pmatrix}=\bB_{\rm s}\begin{pmatrix} \bu(0) & \bu(1) & \cdots & \bu(N)\end{pmatrix}$, and finally we can obtain
\begin{equation*}\hat\bB_{\rm s}=\begin{pmatrix} \hat\bf_{\rm e}(0) & \hat\bf_{\rm e}(1) & \cdots & \hat\bf_{\rm e}(N)\end{pmatrix}\begin{pmatrix} \bu(0) & \bu(1) & \cdots & \bu(N)\end{pmatrix}^{\dag}.\end{equation*}
\item
According to the value of the matrix $\hat\bB_{\rm s}$, we can obtain the estimate of $\bD_{\rm s}$ as $\bD_{\rm s}=\bC_{\rm ac}\bH_3\hat\bB_{\rm s}$, where $\bH_3=\begin{pmatrix} \bzero_{n\times n} & \bI_{n\times n}\end{pmatrix}$.
\end{enumerate}

There is also another kind of way to calculate the estimation of $\hat\bf_{\rm e}(k)$, i.e., first choose one state value from the estimated state sequence $\left\{\bar\bx_{\bu}^{\rm J}(k),0\leqslant k\leqslant N\right\}$ as the initial state value, then based on this initial value, use the state space model
$$\begin{dcases} \bar\bx_{\bu}^{\rm J}(k+1)=\hat\bJ_\bu\bar\bx_{\bu}^{\rm J}(k) \\ \hat\bf_{\rm e}(k)=\hat\bT_{\rm s}(-\hat\bJ_{\rm s}\hat\bX+\hat\bX\hat\bJ_{\bu})\bar\bx_{\bu}^{\rm J}(k)\end{dcases}$$ to calculate $\hat\bf_{\rm e}(k)$ recursively. However, this method may encounter a problem that $\hat\bJ_\bu$ can be unstable, which can lead to the unbounded value of $\hat\bf_{\rm e}(k)$.

Finally, the specific identification algorithm is described in Algorithm \ref{algorithma}, the calculations of $f_{{\rm nat},i}$, $\zeta_i$, $\mathcal{K}$, $\mathcal{C}$, and $\mathcal{B}$ are also included based on using the estimates of $\bar\bA_{\rm a}^{\rm J}$, $\hat\bA_{\rm s}$ and $\hat\bB_{\rm s}$.


\begin{algorithm}[ht] 
\caption{PSSID Algorithm}\label{algorithma}
\SetKwInOut{Input}{Input}
\SetKwInOut{Output}{Output}
\Input{$\left\{\bu(k),0\leqslant k\leqslant N\right\}, \left\{\by(k),0\leqslant k\leqslant N\right\}, \bC_{\rm p}, \bC_{\rm v}, \bC_{\rm ac}, n$}
\Output{$\hat\bA_{\rm s}, \hat\bB_{\rm s}, \hat\bC_{\rm s}, \hat\bD_{\rm s}, \left\{\hat\bf_{\rm e}(k),0\leqslant k\leqslant N\right\}, \hat{f}_{{\rm nat},i}, \hat\zeta_i, \hat{\mathcal{K}}, \hat{\mathcal{D}}, \hat{\mathcal{B}}$}
\tcc{Step I-Augmented model identification}
Compute $\hat\bA_{\rm a}^{\rm d}$ and $\hat\bC_{\rm a}^{\rm d}$ using pure stochastic system identification method \tcp*{See Section \ref{PSSID1}}
$\hat\bA_{\rm a}\gets\frac{1}{T_{\rm s}}\ln(\bA_{\rm a}^{\rm d})$ \tcp*{See (\ref{eq20})}
$\hat\bC_{\rm a}\gets\hat\bC_{\rm a}^{\rm d}$ \tcp*{See (\ref{eq21})}
\tcc{Step II-Real Jordan form and canonical form}
Derive the real Jordan form $(\bA_{\bu}^{\rm J},\bC_{\bu}^{\rm J})$ of the periodic signal model (\ref{eq16}) \tcp*{See (\ref{eq24})}
Derive the real Jordan form of the physical model (\ref{eq3}) \tcp*{See (\ref{eq25})}
Derive an augmented model (\ref{eq26}) in canonical form based on (\ref{eq24}) and (\ref{eq25}) \tcp*{See (\ref{eq26})}
Derive the real Jordan form $(\bA_{\rm J},\bC_{\rm J})$ of the identified model of the augmented system (\ref{eq19a}) \tcp*{See (\ref{eq28}) and (\ref{eq29})}
Derive estimates of modal parameters $f_{{\rm nat},i}$ and $\zeta_i$ of the physical system (\ref{eq3}) using $\bar\bA_{\rm a}^{\rm J}$ \tcp*{See (\ref{eq34}) and (\ref{eq35})}
Derive a canonical form $(\hat\bA_{\rm a}^{\rm J},\hat\bC_{\rm a}^{\rm J})$ of the augmented model $(\bar\bA_{\rm a}^{\rm J},\bar\bC_{\rm a}^{\rm J})$ \tcp*{See (\ref{eq36}) and (\ref{eq37})}
\tcc{Step III-Identification of the state-space model in physical coordinates}
Compute $\hat\bA_{\rm s}$ and $\hat\bC_{\rm s}$ by comparing (\ref{eq26}) with the canonical-form realization $(\hat\bA_{\rm a}^{\rm J},\hat\bC_{\rm a}^{\rm J})$ \tcp*{See (\ref{eq39})-(\ref{eq42})}
Compute $\left\{\hat\bf_{\rm e}(k),0\leqslant k\leqslant N\right\}$, $\hat\bB_{\rm s}$ and $\hat\bD_{\rm s}$ \tcp*{See Step III}
$\hat{\mathcal{K}}\gets\hat\bA_{\rm s}(n+1:2n,1:n)$ \tcp*{See (\ref{eq3})}
$\hat{\mathcal{C}}\gets\hat\bA_{\rm s}(n+1:2n,n+1:2n)$ \tcp*{See (\ref{eq3})}
$\hat{\mathcal{B}}\gets\hat\bB_{\rm s}(n+1:2n,1:r)$ \tcp*{See (\ref{eq3})}
\textbf{Return} $\hat\bA_{\rm s}, \hat\bB_{\rm s}, \hat\bC_{\rm s}, \hat\bD_{\rm s}, \left\{\hat\bf_{\rm e}(k),0\leqslant k\leqslant N\right\}, \hat{f}_{{\rm nat},i}, \hat\zeta_i, \hat{\mathcal{K}}, \hat{\mathcal{D}}$, and $\hat{\mathcal{B}}$.
\end{algorithm}

At the end of this section, three points needs to be clarified for stochastic subspace identification (the ``stochastic algorithm $3$''):
\begin{enumerate}[{(1)}]
\item
Consistency: A steady-state solution of the Riccati equation used for the finite interval of the measurement data may lead to biased covariance sequence estimation \cite{Verhaegen2007}.
\item
Stability: It can be seen that the eigenvalues of the signal model (\ref{eq16}) are close to the bound of the unit circle, so the identified model can be unstable, however, we do not need to deal with the stability of the identified system matrix of the signal model, while we only need to take an action to enforce the system (\ref{eq3}) to be stable \cite{Ljung2016}.
\item
Causality: As described in \cite{Qin2006} and \cite{Qin2005}, there will be a causality problem if the combined deterministic-stochastic identification methods proposed in \cite{Overschee1996} are used, even though there is no problem from a consistency point of view. In our paper, we use the pure stochastic identification method without the effect from the deterministic input, so we do not need to consider about the causality problem.
\end{enumerate}

\subsection{Blind identification}\label{BLSID}
In practice, finite-length input and output signals are collected for offline system identification, and sometimes we cannot choose the type of the input signal freely, as a result the input signal can be a non-periodic signal, so the proposed identification method above based on using periodic input signal can not be used directly. However, it is known that every non-periodic signal can be seen as a periodic signal with an infinite period. Based on this knowledge, a proposition (Proposition \ref{propositiona}) for the finite-length band-limited signal modeling is first given, and then a numerical example is used to demonstrate the effectiveness of the finite-length signal modeling. The finite-length band-limited signal can belong to a part of either a periodic signal or a part of non-periodic signal, so the proposed identification method in Section \ref{PSSID} could still be implemented. Following Proposition \ref{propositiona}, another proposition about the signal model order is introduced. As aforementioned, in practice sometimes we may not know both type (e.g., periodic and non-periodic) and value of the input signal, based on this consideration, finally a blind identification method is derived by modifying the proposed identification method in Section \ref{PSSID}.

\subsubsection{Finite-length signal modeling}
\begin{proposition}\label{propositiona}
For a discrete-time, length-limited, and band-limited signal $v(k)\in\mathbb{R}$, it can be extended into a periodic signal $v_{\rm p}(k)\in\mathbb{R}$, and the period of the expanded periodic signal $v_{\rm p}(k)$ is $NT_{\rm s}$, where $N$ is the sampling points number of $v(k)$, and $T_{\rm s}$ denotes the sampling time. If the continuous-time version of the signal $v_{\rm p}(k)$ is a continuous function and satisfies the Dirichlet conditions in the interval $[0,NT_{\rm s})$, the periodic signal $v_{\rm p}(k)$ can be expressed as a Fourier series in the interval $[0,NT_{\rm s})$, that is
\begin{align}
v_{\rm p}(k)&=\frac{a_0}{2}+\sum_{n=1}^{\infty}(a_n\cos(nwkT_{\rm s})+b_n\sin(nwkT_{\rm s})) \nonumber \\
&=\alpha_0+\sum_{n=1}^{\infty}\alpha_n\sin(nwkT_{\rm s}+\beta_n),\label{eq62}
\end{align}
where $w=\frac{2\pi}{NT_{\rm s}}$, and the coefficients $a_0=\frac{2}{NT_{\rm s}}\int_{-\frac{NT_{\rm s}}{2}}^{\frac{NT_{\rm s}}{2}}v_{\rm p}(t)\, \mathrm{d}t$, $a_n=\frac{2}{NT_{\rm s}}\int_{-\frac{NT_{\rm s}}{2}}^{\frac{NT_{\rm s}}{2}}v_{\rm p}(t)\cos(nwt)\, \mathrm{d}t$, and $b_n=\frac{2}{NT_{\rm s}}\int_{-\frac{NT_{\rm s}}{2}}^{\frac{NT_{\rm s}}{2}}v_{\rm p}(t)\sin(nwt)\, \mathrm{d}t$.
\end{proposition}

Since the result of Proposition \ref{propositiona} is obvious, the proof of it is not given here. Based on Proposition \ref{propositiona} and (\ref{eq62}), in the interval $[0,NT_{\rm s})$, the length-limited and band-limited signal $v(k)$ can be approximately modeled as the output of a linear autonomous state-space model like the periodic signal modeling described in Section \ref{periodic signal model}, that is
\begin{equation}\label{eq63}
\begin{dcases}
\bx_v(k+1)=\bA_v^{\rm d}\bx_v(k), \\
v(k)=\bC_v^{\rm d}\bx_v(k)+e_v(k),
\end{dcases}
\end{equation}
where $e_v(k)$ denotes the part of $v(k)$ not considered, for $0\leqslant k\leqslant NT_{\rm s}$.

The matrices $\bA_v^{\rm d}$ and $\bC_v^{\rm d}$ of the system (\ref{eq63}) are given as
\begin{equation*}
\bA_v^{\rm{d}}=
\begin{pmatrix}
1 & \bzero & \cdots & \bzero \\
\bzero & \bA_{v,\rm{1}}^{\rm{d}} & \ddots & \vdots \\
\vdots & \ddots & \ddots & \bzero \\
\bzero & \cdots & \bzero & \bA_{v,n_{\rm{a}}-1}^{\rm{d}} \\
\end{pmatrix}
\end{equation*}
and
\begin{equation*}
\bC_v^{\rm{d}}=
\begin{pmatrix}
1 & \bC_{v,\rm{1}}^{\rm{d}} & \cdots & \bC_{v,n_{\rm{a}}-1}^{\rm{d}} \\
\end{pmatrix}.
\end{equation*}
The individual block entries in these block matrices can be represented as
\begin{equation*}
\bA_{v,i}^{\rm{d}}=
\begin{pmatrix}
\cos(i2\pi \frac{1}{N}) & \sin(i2\pi \frac{1}{N})\vspace{1ex} \\
-\sin(i2\pi \frac{1}{N}) & \cos(i2\pi \frac{1}{N}) \\
\end{pmatrix}
\end{equation*}
and
\begin{equation*}
\bC_{v,i}^{\rm{d}}=
\begin{pmatrix}
1 & 0 \\
\end{pmatrix}.
\end{equation*}

\begin{remark}
For the length-limited signal $v(k)$, it can be only represented as the output of the system (\ref{eq63}) in the interval $[0,NT_{\rm s})$.
\end{remark}

Based on Proposition \ref{propositiona} and the signal model (\ref{eq63}), we can know that for a length-limited and non-periodic signal, we can also regard it as a period of a periodic signal with the period $NT_{\rm s}$. This is suitable for practical application, since in practice the input and the output can be non-periodical. Even though in theory Proposition \ref{propositiona} and the modeling of the non-periodic signal are feasible, the effectiveness of them should be checked, one way is to estimate the non-periodic signal using the model (\ref{eq63}) and a given model, for the specific operation we can use the following example to illustrate.

\begin{exmp}\label{exampa}
Given a non-periodic and band-limited signal $v(k)\in\mathbb{R}$ with length $N=60,000$, then extract one part $v_1(k)$ with length $N_1=10,000$ from $v(k)$, these two signals are shown in Figure \ref{fig3}. Use $v_1(k)$ as an input to excite an observable system (\ref{eq64}) (Its Bode plot is shown in Figure \ref{fig4}.) to obtain a simulated output $y(k)$ with length $N$, and then use the output $y(k)$ to estimate the input signal $v_1(k)$ based on using the method described in Figure \ref{fig5}, it can also be referred to \cite{Han2018a} to see more details about the method, and finally the estimation results of $v_1(k)$ are shown in Figure \ref{fig6} and Figure \ref{fig7}. The sampling frequency is chosen as $10,000$ Hz, and the observable system is given as follows
\begin{equation}\label{eq64}
\begin{dcases}
\bx(k+1)=\bA\bx(k)+\bB v_1(k)+\bG e(k), \\
y(k)=\bC\bx(k)+e(k),
\end{dcases}
\end{equation}
where $\bA=\begin{pmatrix}0.5 & 0\\0 & 0.6\end{pmatrix}$, $\bB=\begin{pmatrix}1 \\ 0.5\end{pmatrix}$, $\bC=\begin{pmatrix}1 & 1 \end{pmatrix}$, and $\bG=\begin{pmatrix}1 \\0.5 \end{pmatrix}$. $\{e(k)\}$ is zero-mean white noise sequence with covariance $0.0001$.
\end{exmp}

\begin{figure}[ht]
\centering
\includegraphics[scale=0.2,angle=0]{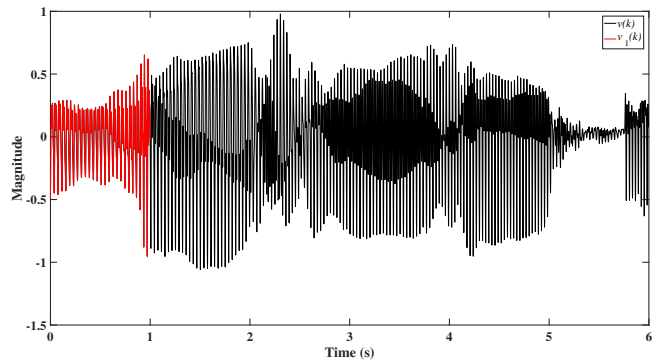}
\caption{Non-periodic signal and one of its part.}
\label{fig3}
\end{figure}

\begin{figure}[ht]
\centering
\includegraphics[scale=0.2,angle=0]{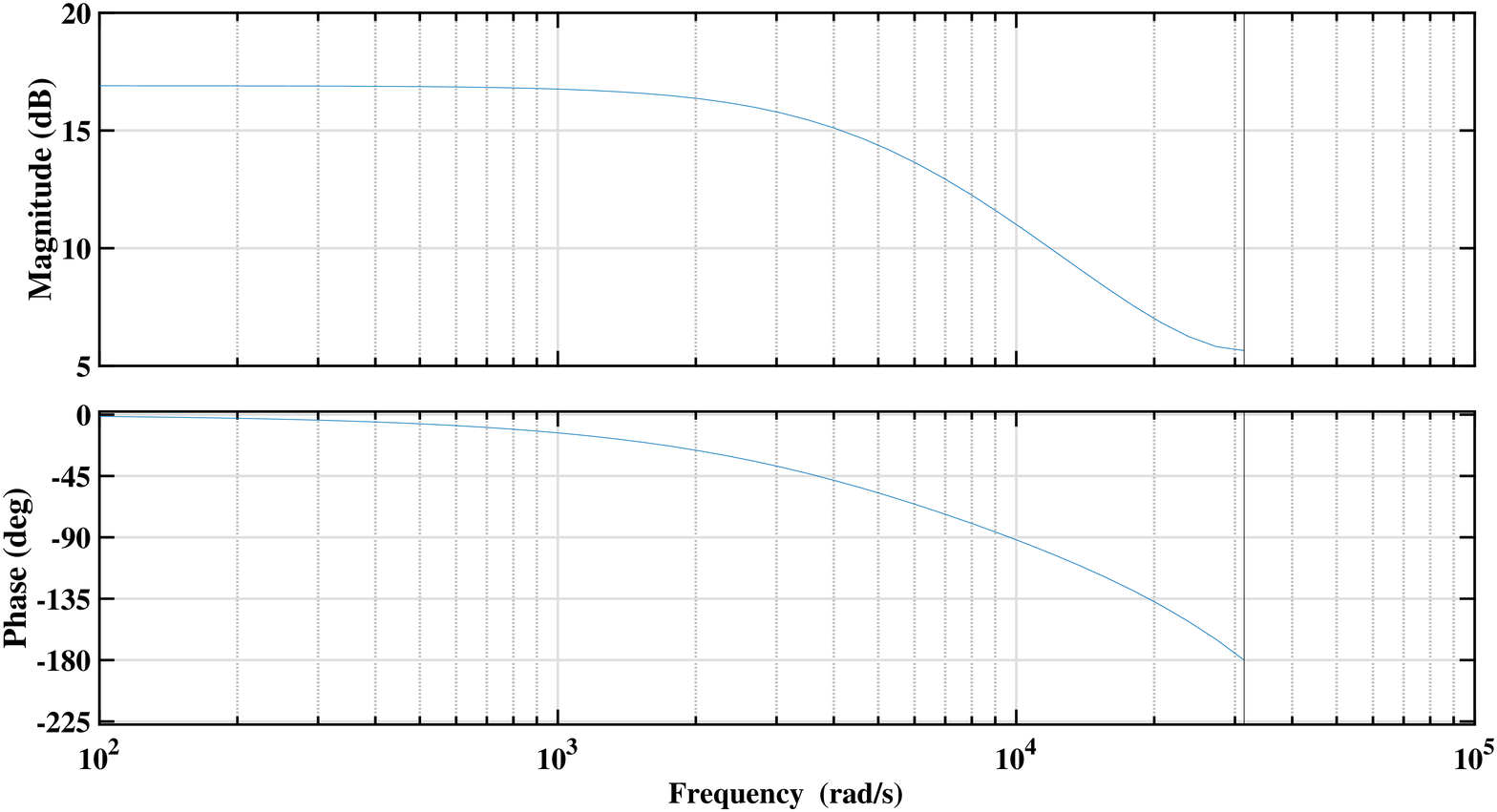}
\caption{Bode plot of the system (\ref{eq64}).}
\label{fig4}
\end{figure}

\begin{figure}[ht]
\centering
\includegraphics[scale=0.5,angle=0]{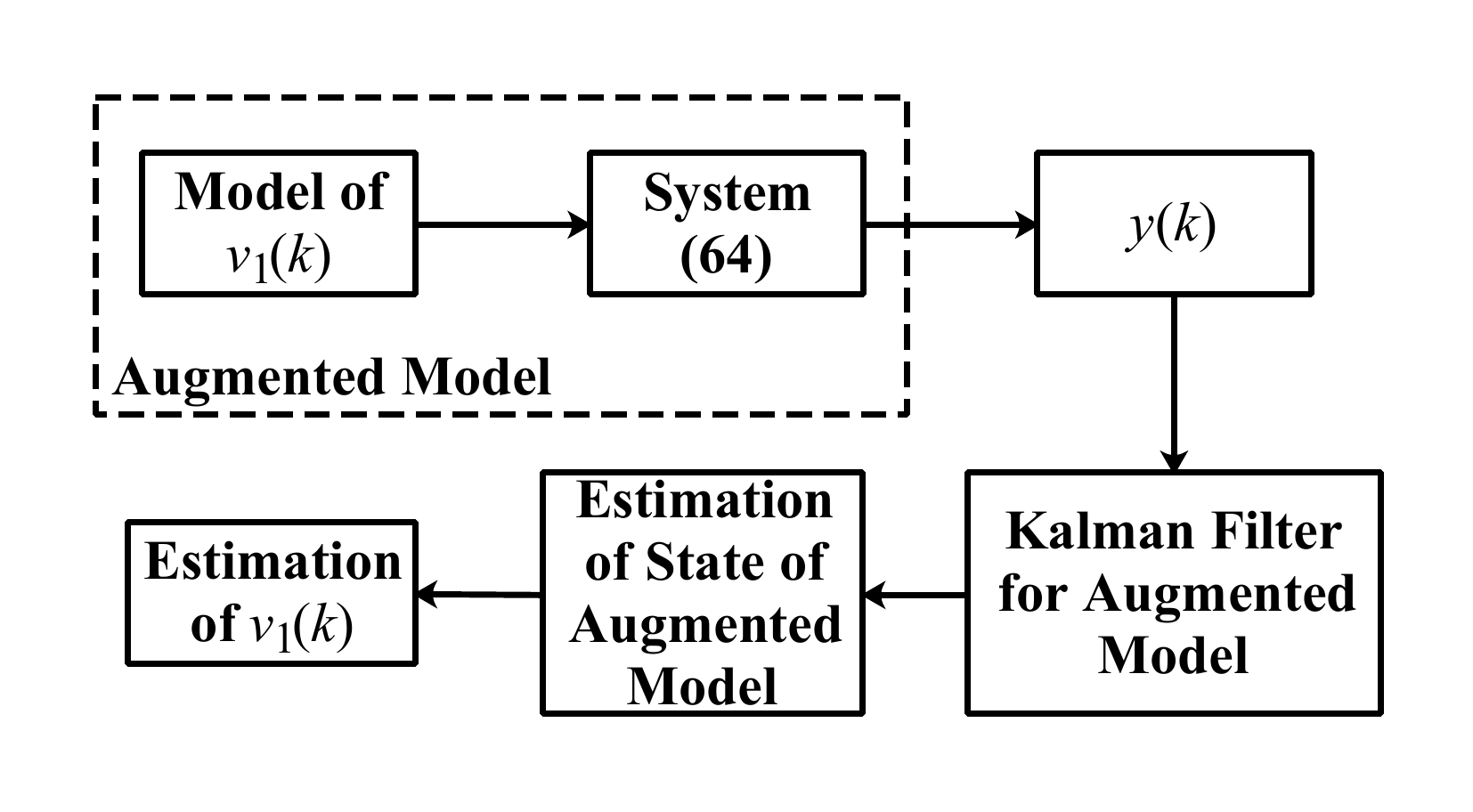}
\caption{Block diagram for estimation method.}
\label{fig5}
\end{figure}

\begin{figure}[ht]
\centering
\includegraphics[scale=0.2,angle=0]{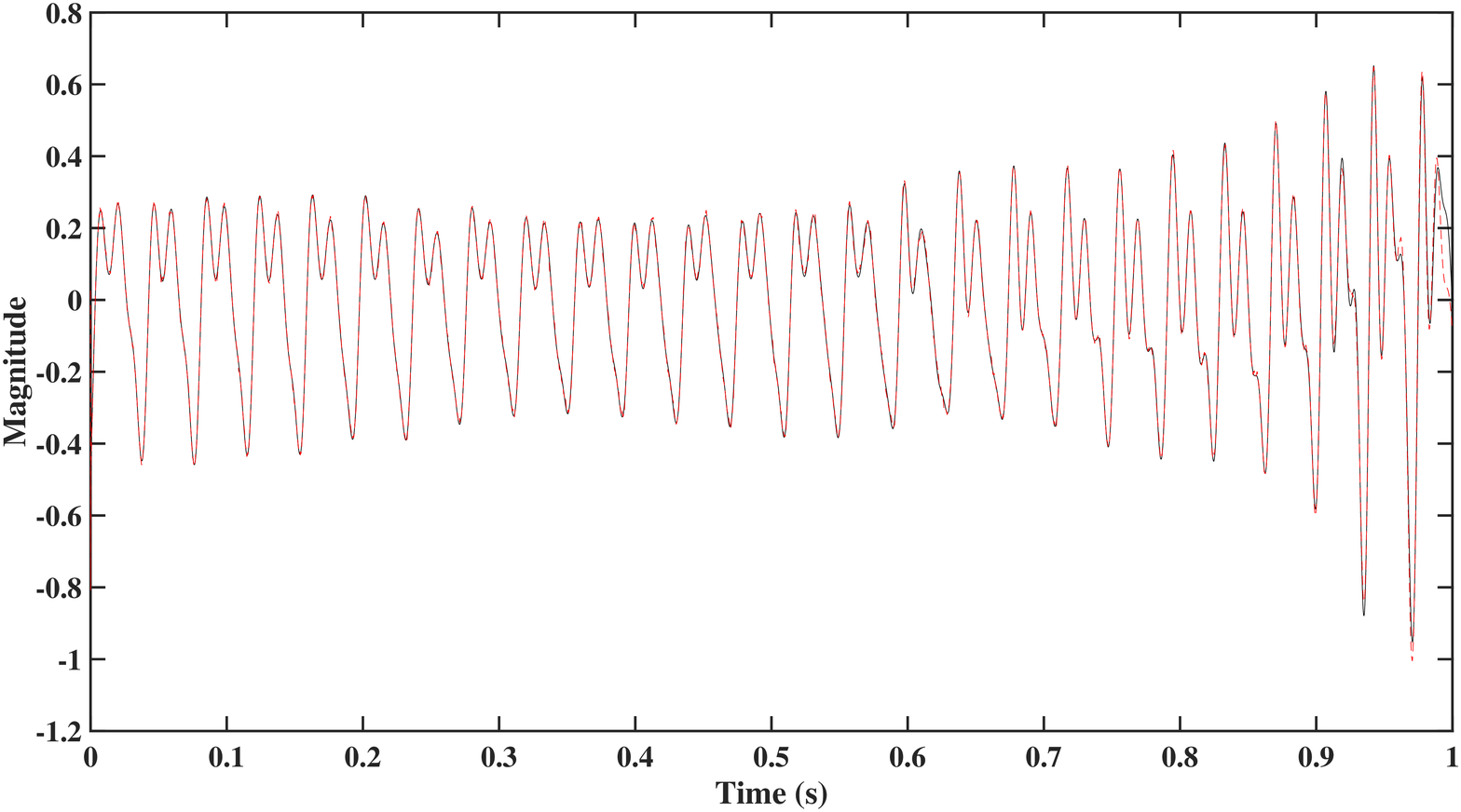}
\caption{True value (in black) and estimation value (in red) in time domain.}
\label{fig6}
\end{figure}

\begin{figure}[ht]
\centering
\includegraphics[scale=0.2,angle=0]{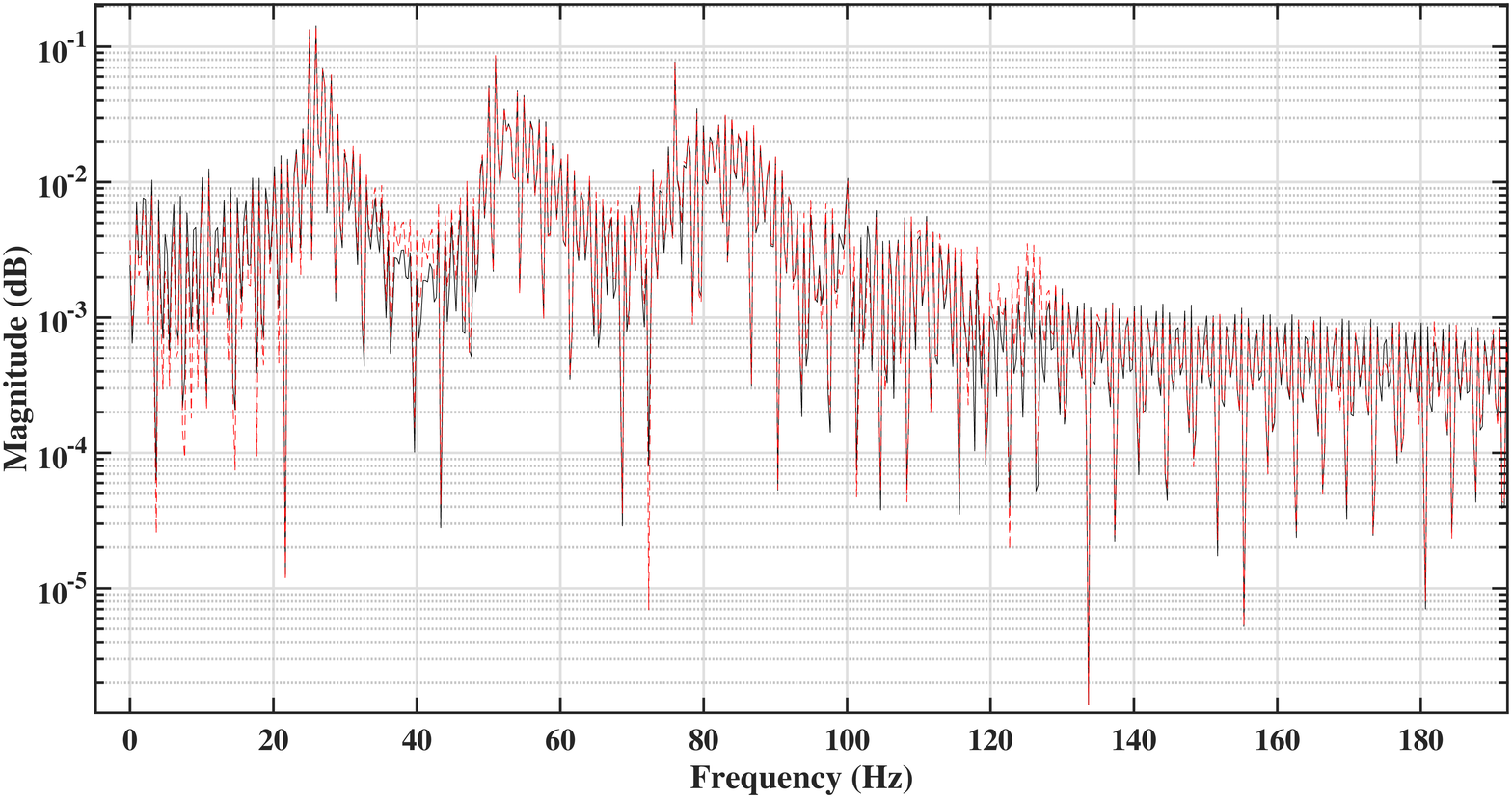}
\caption{True value (in black) and estimation value (in red) in frequency domain.}
\label{fig7}
\end{figure}

According to the estimation result from Example \ref{exampa}, it can show the effectiveness of Proposition \ref{propositiona} for a length-limited non-periodic signal estimation. As a comparison, for an infinite-length non-periodic signal , the modeling of it can be described in the following proposition.

\begin{proposition}\label{propositionb}
Given an infinite-length non-periodic signal $u(t)\in\mathbb{R}$ which is continuous and satisfies Dirichlet conditions in any limited interval and absolutely integrable in the interval $(-\infty,+\infty)$, it cannot be described as the output of a linear time-invariant autonomous state-space model with a fixed model order except one case that the non-periodic signal $u(t)$ is a sum of periodic signals. While for the discrete-time version $u(k)$, in any case it cannot be described as the output of a linear time-invariant autonomous state-space model (LTIASSM) with a fixed model order.
\end{proposition}
\begin{proof}
Under the preconditions mentioned in Proposition \ref{propositionb}, we use the following three aspects to prove this proposition:
\begin{enumerate}[{(1)}]
\item
When the continuous-time $u(t)$ is not a sum of periodic signals, according to Proposition \ref{propositiona}, under some preconditions, for a length-limited signal, it can be represented as the output of an autonomous state-space model in a limited time interval, i.e., for different time lengths, the model orders of the autonomous state-space models are different, and these different state-space models cannot transform to each other using a similar transformation matrix or model reduction techniques, so in this aspect, $u(t)$ can not be described as the output of a LTIASSM with a fixed model order.
\item
When the continuous-time $u(t)$ is a sum of periodic signals, i.e., $u(t)=x_1(t)+x_2(t)+\ldots+x_N(t)$, where $x_i(t)$ for $1\leqslant i\leqslant N$ is a periodic signal with fundamental period $T_i$ seconds, and not every one of the ratios $\frac{T_1}{T_i}$ for $i=2,3,\ldots,N$ is a rational number. Since $u(t)$ is a sum of periodic signal, in addition every periodic signal can be modeled as the output of a LTIASSM with a fixed model order (see Section \ref{periodic signal model}), $u(t)$ can be also described as the output of a LTIASSM with a fixed model order.
\item
The discrete-time $u(k)$ cannot be described as the output of a LTIASSM with a fixed model order \cite{Palacios1998}.
\end{enumerate}
\end{proof}

\begin{remark}\label{remarka}
According to Proposition \ref{propositiona} and Proposition \ref{propositionb}, for a length-limited and band-limited signal $v(t)$ in continuous time (The signal $v(t)$ can belong to a part of either a periodic signal or a non-periodic signal), if it is continuous and satisfies Dirichlet conditions in the interval $[0,L)$, where $L$ is the length of $v(t)$, the discrete-time signal $v(k)$ can be denoted as the output of a LTIASSM with a fixed model order, and the model order size depends on both length of the signal $v(k)$ and the maximum frequency of $v(k)$. Furthermore, if $v(k)$ belongs to a part of a periodic signal, the model order size of the model of $v(k)$ can also be chosen as a fixed number not changing with the length of $v(k)$, however for this case we should know the fundamental period of the period signal which contains the part $v(k)$, which is not possible sometimes since we may not know the type of the collected signal.
\end{remark}

\subsubsection{Blind identification algorithm}
Based on the above analysis, the proposed system identification method in Section \ref{PSSID} can be implemented under either periodic excitations or non-periodic excitations, and the input signal model order is fixed up to a specific signal length, maybe the collected input signal belongs to a part of a periodic signal of which the model order does not change with the signal length, which is mentioned in Remark \ref{remarka}, but in practical application, this information is usually not known, so in this section we assume that we do not know the type of the input signal, and use Proposition \ref{propositiona} to build a model for it. Furthermore, in practice, e.g., in structure area, usually the input force can not be measured, so a blind identification method should be considered and involved. An important and critical step prior to the design (and use) of the blind identification of the physical model (\ref{eq3}) is to check the observability of the augmented model (\ref{eq19}) since subspace identification is involved in the blind identification. In \ref{appendixb}, the observability is studied and checked. In this section, we temporarily assume that the augmented model (\ref{eq19a}) is observable and derive a blind identification algorithm based on the proposed identification method in Section \ref{PSSID}. The blind identification algorithm using only a finite-length output signal is summarized in the following.

\begin{algorithm}[ht]
\caption{Blind Identification Algorithm}\label{algorithmb}
\SetKwInOut{Input}{Input}
\SetKwInOut{Output}{Output}
\Input{$\left\{\bu(k),0\leqslant k\leqslant N\right\}, \left\{\by(k),0\leqslant k\leqslant N\right\}, \bC_{\rm p}, \bC_{\rm v}, \bC_{\rm ac}, n$}
\Output{$\hat\bA_{\rm s}, \hat\bC_{\rm s}, \left\{\hat\bf_{\rm e}(k),0\leqslant k\leqslant N\right\}, \hat{f}_{{\rm nat},i}, \hat\zeta_i, \hat{\mathcal{K}}, \hat{\mathcal{D}}$}
\tcc{Step I-Augmented model identification}
Compute $\hat\bA_{\rm a}^{\rm d}$ and $\hat\bC_{\rm a}^{\rm d}$ using pure stochastic system identification method \tcp*{See Section \ref{PSSID1}}
$\hat\bA_{\rm a}\gets\frac{1}{T_{\rm s}}\ln(\bA_{\rm a}^{\rm d})$ \tcp*{See (\ref{eq20})}
$\hat\bC_{\rm a}\gets\hat\bC_{\rm a}^{\rm d}$ \tcp*{See (\ref{eq21})}
\tcc{Step II-Real Jordan form and canonical form}
Derive the real Jordan form $(\bA_{\bu}^{\rm J},\bC_{\bu}^{\rm J})$ of the periodic signal model (\ref{eq16}) \tcp*{See (\ref{eq24})}
Derive the real Jordan form of the physical model (\ref{eq3}) \tcp*{See (\ref{eq25})}
Derive an augmented model (\ref{eq26}) in canonical form based on (\ref{eq24}) and (\ref{eq25}) \tcp*{See (\ref{eq26})}
Derive the real Jordan form $(\bA_{\rm J},\bC_{\rm J})$ of the identified model of the augmented system (\ref{eq19a}) \tcp*{See (\ref{eq28}) and (\ref{eq29})}
Derive estimates of modal parameters $f_{{\rm nat},i}$ and $\zeta_i$ of the physical system (\ref{eq3}) using $\bar\bA_{\rm a}^{\rm J}$ \tcp*{See (\ref{eq34}) and (\ref{eq35})}
Derive a canonical form $(\hat\bA_{\rm a}^{\rm J},\hat\bC_{\rm a}^{\rm J})$ of the augmented model $(\bar\bA_{\rm a}^{\rm J},\bar\bC_{\rm a}^{\rm J})$ \tcp*{See (\ref{eq36}) and (\ref{eq37})}
\tcc{Step III-Identification of the state-space model in physical coordinates}
Compute $\hat\bA_{\rm s}$ and $\hat\bC_{\rm s}$ by comparing (\ref{eq26}) with the canonical-form realization $(\hat\bA_{\rm a}^{\rm J},\hat\bC_{\rm a}^{\rm J})$ \tcp*{See (\ref{eq39})-(\ref{eq42})}
Compute $\left\{\hat\bf_{\rm e}(k),0\leqslant k\leqslant N\right\}$ \tcp*{See Step III}
$\hat{\mathcal{K}}\gets\hat\bA_{\rm s}(n+1:2n,1:n)$ \tcp*{See (\ref{eq3})}
$\hat{\mathcal{C}}\gets\hat\bA_{\rm s}(n+1:2n,n+1:2n)$ \tcp*{See (\ref{eq3})}
\textbf{Return} $\hat\bA_{\rm s}, \hat\bC_{\rm s}, \left\{\hat\bf_{\rm e}(k),0\leqslant k\leqslant N\right\}, \hat{f}_{{\rm nat},i}, \hat\zeta_i, \hat{\mathcal{K}}, \hat{\mathcal{D}}$.
\end{algorithm}

It should be noted that similar to the proposed identification method in Section \ref{PSSID}, the forepart of the estimation sequence $\left\{\hat\bx_{\rm a}^*(k),0\leqslant k\leqslant N\right\}$ can be inaccurate caused by the Kalman filtering when the initial value is not chosen well.

The last important point in this section is that in practical identification, the length of the output signal should be chosen enough for asymptotical identification purpose, however, if the length is large, according to Proposition \ref{propositiona} and Remark \ref{remarka}, the value of $NT_{\rm s}$ can then be large, as a result the dimension of the signal model can become large, which is not suitable for practical application, so a method should be taken to solve this problem. Figure \ref{fig8} gives us a possible way to solve this problem.

\begin{figure}[ht]
\centering
\includegraphics[scale=0.5,angle=0]{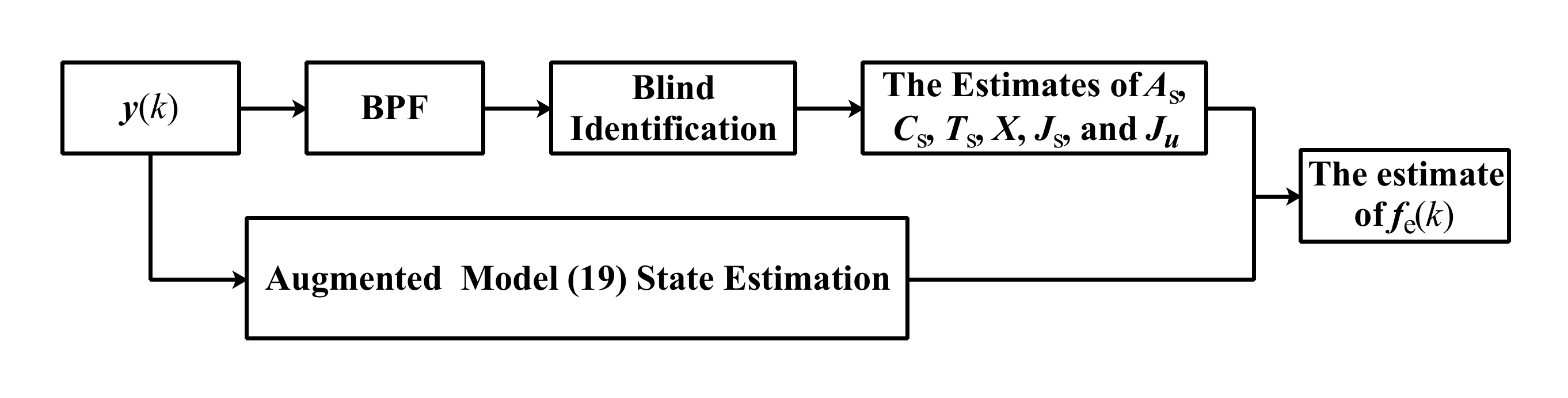}
\caption{Block diagrams for blind identification implementation. (BPF=band-pass filter.)}
\label{fig8}
\end{figure}

As shown in Figure \ref{fig8}, first a band-pass filter is used to decrease the frequency bandwidth of the output signal, i.e., to decrease frequency components of the output signal, which is equivalent to decreasing the size of the number $n_{\rm a}$. (The reason why they are equivalent is explained in the next paragraph.) But the persistent excitation condition in Definition \ref{definitiona} should be satisfied. Then the filtered output signal is used to identify the system according to the blind identification method proposed in this section. Finally based on using the state estimated using the original output signal $\by(k)$ and the Kalman filter for the augmented model (\ref{eq19}), the estimation sequence $\left\{\bar\bx_{\bu}^{\rm J}(k),0\leqslant k\leqslant N\right\}$ can be obtained, and afterwards the effective input sequence $\left\{\hat\bf_{\rm e}(k),0\leqslant k\leqslant N\right\}$ can be obtained. In practice, we should consider the values of both signal length and frequency bandwidth of the input signals, tune them together to get a small $n_{\rm a}$.

In the above, we say we tune the bandwidth of the input signal, which sounds a bit strange since the input signal is unknown. Actually, we can only change the bandwidth of the collected output signal using a band-pass filter. According to example \ref{exampa}, we can know that the estimated input signal’s frequency components are directly related with the output signal’s frequency components, i.e., if some frequency components do not occur in the output, we could not find the same frequency components in the estimated input. So based on the above description, tuning the input bandwidth means tuning the bandwidth of the output signal.

There are three important contributions in this section, the first one is that for non-periodic input signals, the proposed identification method in Section \ref{PSSID} can still be used, which is more practical since in practice most signals are non-periodic signals. The second one is that for input signal with unknown type and value, a blind identification method is derived. The third one is that a framework for practical use of the proposed blind identification method is derived to solve the problems caused by both signal length and frequency bandwidth of the input signal.

\section{Numerical examples}\label{s5}
To evaluate the proposed identification methods in Section \ref{PSSID} and Section \ref{BLSID}, an n-degree-of-freedom linear mass-spring-damper model as shown in Figure \ref{fig9} is used in this Section to check the effectiveness of the proposed identification methods. It represents a multi-story building with pure sheer if the last spring-damper $(k_{n+1},d_{n+1})$ is not considered. In Figure \ref{fig9}, $A_i$, $V_i$, and $P_i$ stand for the measurements of acceleration, velocity, and position of $i$-th block, respectively. As assumed in Assumption \ref{assum2}, we only use the acceleration sensor in this paper.

\begin{figure}[ht]
\centering
\includegraphics[scale=0.2,angle=0]{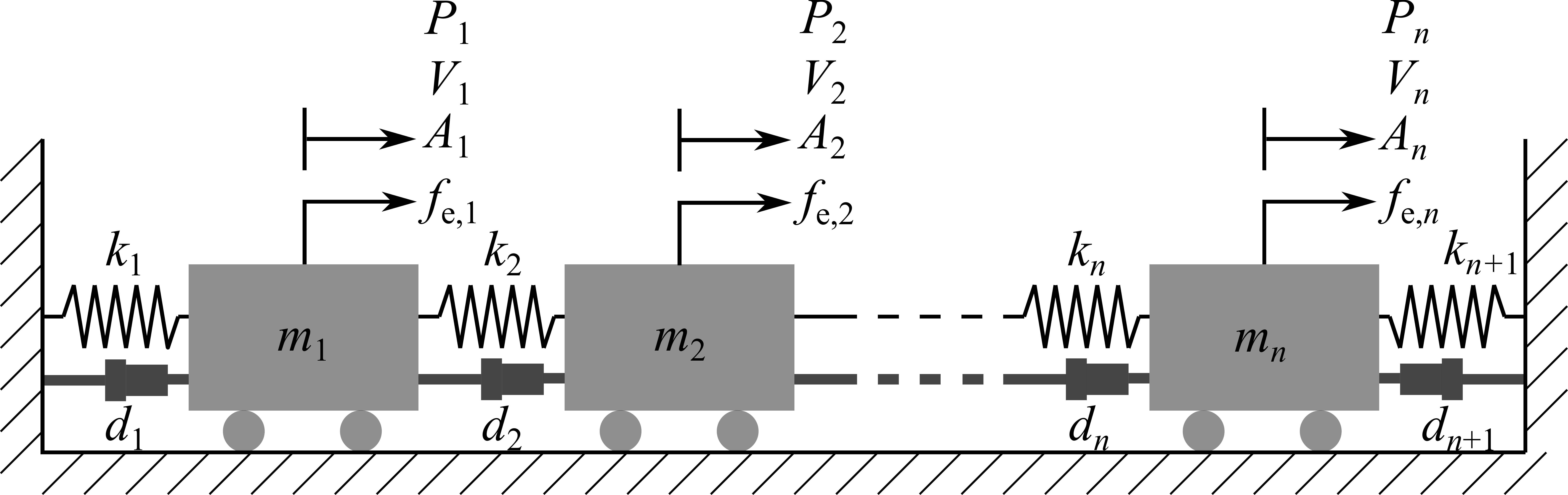}
\setlength{\abovecaptionskip}{0pt}
\setlength{\belowcaptionskip}{0pt}
\caption{Schematic of the mass-spring-damper model.}
\label{fig9}
\end{figure}

According to the system shown in Figure \ref{fig9}, the system can be governed by the following system of simultaneous equations:

\begin{align*}
&m_1A_1(t)+(d_1+d_2)V_1(t)-d_2V_2(t)+(k_1+k_2)P_1(t)-k_2P_2(t)=f_{{\rm e},1}(t), \nonumber \\
&m_2A_2(t)-d_2V_1(t)+(d_2+d_3)V_2(t)-d_3V_3(t)-k_2P_1(t)+(k_2+k_3)P_2(t)-k_3P_3(t)=f_{{\rm e},2}(t), \nonumber \\
&\vdots \nonumber \\
&m_nA_n(t)-d_nV_{n-1}(t)+(d_n+d_{n+1})V_n(t)-d_nP_{n-1}(t)+(k_n+k_{n+1})P_n(t)=f_{{\rm e},n}(t), \nonumber \\
\end{align*}
which can be condensed in a more compact form (which is the same as (\ref{eq1})), as

\begin{equation}
\bM\ddot{\bq}(t)+\bD\dot{\bq}(t)+\bK\bq(t)=\bf_{\rm e}(t),
\end{equation}
where
\[\bq(t)=
\begin{pmatrix}
P_1(t) \\
P_2(t) \\
\vdots \\
P_n(t)
\end{pmatrix},
\dot\bq(t)=
\begin{pmatrix}
V_1(t) \\
V_2(t) \\
\vdots \\
V_n(t)
\end{pmatrix},
\ddot\bq=
\begin{pmatrix}
A_1(t) \\
A_2(t) \\
\vdots \\
A_n(t)
\end{pmatrix},\]

\begin{equation}
\bf_{\rm e}(t)=
\begin{pmatrix}
f_{{\rm e},1}(t) \\
f_{{\rm e},2}(t) \\
\vdots \\
f_{{\rm e},n}(t)
\end{pmatrix},
\end{equation}
and the matrices $\bM$, $\bD$, and $\bK$ are symmetric, i.e.,
\begin{equation}\label{eq65}
\bM=
\begin{pmatrix}
m_1 & 0 & 0 & \cdots & 0 \\
0 & m_2 & 0 & \cdots & 0 \\
0 & 0 & \ddots & \ddots & \vdots \\
\vdots & \vdots & \ddots & m_{n-1} & 0 \\
0 & 0 & \cdots & 0 & m_n \\
\end{pmatrix},
\end{equation}

\begin{equation}\label{eq66}
\bD=
\begin{pmatrix}
d_1+d_2 & -d_2 & 0 & \cdots & 0 \\
-d_2 & d_2+d_3 & -d_3 & \ddots & \vdots \\
0 & -d_3 & \ddots & \ddots & 0 \\
\vdots & \ddots & \ddots & d_{n-1}+d_{n} & -d_n \\
0 & \cdots & 0 & -d_n & d_n+d_{n+1} \\
\end{pmatrix},
\end{equation}
and
\begin{equation}\label{eq67}
\bK=
\begin{pmatrix}
k_1+k_2 & -k_2 & 0 & \cdots & 0 \\
-k_2 & k_2+k_3 & -k_3 & \ddots & \vdots \\
0 & -k_3 & \ddots & \ddots & 0 \\
\vdots & \ddots & \ddots & k_{n-1}+k_{n} & -k_n \\
0 & \cdots & 0 & -k_n & k_n+k_{n+1} \\
\end{pmatrix}.
\end{equation}

Here, it should be noted that if the last spring-damper is not considered the values of $d_{n+1}$ and $k_{n+1}$ are zero.

Additionally, we make the following assumption for the system (\ref{eq65}):
\begin{assumption}
The viscous damping matrix $\bD$ is a linear combination of both mass matrix and stiffness matrix, i.e., $\bD$ is a proportional damping (or say Rayleigh damping) and $\bD=\varepsilon\bM+\nu\bK$, where $\varepsilon$ and $\nu$ are constants.
\end{assumption}

Based on the system shown in Figure \ref{fig9}, two numerical cases are studied in this section to validate the proposed identification methods:

\begin{enumerate}[{(1)}]
\item
Different DOFs: in this case, the DOF of the system (\ref{eq65}) is chosen as $3$, $4$, $5$, $6$, $7$, and $8$, respectively, to validate the proposed methods in Section \ref{PSSID} (PSSID) and \ref{BLSID} (blind identification).
\item
Different signal-to-noise ratios (SNRs): in this case, six different SNRs for both input signal and output signal are used to validate the proposed methods.
\end{enumerate}


\section{Practical example}\label{s6}

\section{Discussions and conclusions}\label{s7}


\appendix
\section{Real Jordan form}\label{appendixa}
In this section, prior to introducing the basics of the real Jordan form of a given square matrix, which are used in both Section \ref{PSSID} and Section \ref{BLSID}, the real modal decomposition (a kind of canonical form for matrices with non-repeated eigenvalues) of a two-by-two matrix is considered first for the sake of simplicity.

It is known that if $\lambda=\sigma+j\omega$ is a complex eigenvalue of a real matrix $\bA_{2\times2}\in \mathbb{R}^{2\times2}$, its conjugate $\lambda^{*}$ is also an eigenvalue, and their corresponding eigenvectors are a pair of complex conjugate vectors $(\balpha\pm j\bbeta)$. Therefore, the eigen decomposition of such matrix can be written as
\begin{equation}\label{A1}
\bA_{2\times2} = V\Lambda V^{-1},
\end{equation}
where $$\bV=\begin{pmatrix} \balpha+j\bbeta & \balpha-j\bbeta\end{pmatrix},$$
$$\bLambda=\begin{pmatrix} \sigma+j\omega & 0 \\ 0 & \sigma-j\omega\end{pmatrix},$$
and
$$\bV^{-1}=\begin{pmatrix} \bphi^{\rm T}-j\bpsi^{\rm T} \\ \bphi^{\rm T}+j\bpsi^{\rm T}\end{pmatrix}.$$

So
\begin{equation}\label{A2}
\bA_{2\times2}=2(\sigma\balpha-\omega\bbeta)\bphi^{\rm T}+2(\omega\balpha+\sigma\bbeta)\bpsi^{\rm T}.
\end{equation}

Equation (\ref{A2}) can be also written in the following form that is a real modal decomposition of the matrix $\bA_{\rm m}$, that is
\begin{equation}\label{A3}
\bA_{2\times2} = \bV_{\rm D}\bLambda_{\rm D}\bV_{\rm D}^{-1},
\end{equation}
where $$\bV_{\rm D}=\begin{pmatrix} \sqrt{2}\balpha & \sqrt{2}\bbeta\end{pmatrix},$$
$$\bLambda_{\rm D}=\begin{pmatrix} \sigma & \omega \\ -\omega & \sigma\end{pmatrix},$$
and
$$\bV_{\rm D}^{-1}=\begin{pmatrix} \sqrt{2}\bphi^{\rm T} \\ \sqrt{2}\bpsi^{\rm T}\end{pmatrix}.$$

To obtain the real modal decomposition for larger matrices with unique eigenvalues, the same process can be adopted for complex eigenvalues such that every complex pair of eigenvalues/eigenvectors is substituted with their corresponding real alternatives as described for a two-by-two matrix. Adding the real eigenvalues/eigenvectors to the end of decomposition, a block-diagonalized real modal form would be obtained.

In the case of repeated eigenvalues, the decomposition is slightly different and called real Jordan decomposition. The real Jordan form (denoted as $\bJ$) of an arbitrary matrix $\bA_{N\times N}\in\mathbb{R}^{N\times N}$ with $M(\leqslant N)$ distinct eigenvalues $\lambda_{i}\,(i=1,2,\ldots,M)$, each of which is repeated $s_{i}$ times, is defined by
\begin{equation}\label{A4}
\bA_{N\times N}=\bP\bJ\bP^{-1},
\end{equation}
where the matrix $\bP$ is the eigenvector matrix, and the matrix $\bJ$ is
\begin{equation*}
\bJ=
    \begin{pmatrix}
        \bJ^{(1)} & \bzero & \cdots & \bzero \\
        \bzero & \bJ^{(2)} & \cdots & \bzero \\
        \vdots & \vdots & \ddots & \vdots \\
        \bzero & \bzero & \cdots & \bJ^{(M)} \\
    \end{pmatrix}.
\end{equation*}

Each Jordan block in $\bJ$, i.e., $\bJ^{(i)}$ corresponds to a distinct eigenvalue $\lambda_{i}$ which is repeated $s_{i}$ times and is defined by the following
\begin{equation*}
\bJ^{(i)}=
    \begin{pmatrix}
        \bC_i & \bI_i & \bzero & \cdots & \bzero \\
        \bzero & \bC_i & \bI_i & \cdots & \bzero \\
        \vdots & \vdots & \ddots & \ddots & \vdots \\
        \bzero & \bzero & \bzero & \bC_i & \bI_i \\
        \bzero & \bzero & \bzero & \bzero & \bC_i\\
    \end{pmatrix}\in\mathbb{R}^{s_i\times s_i},
\end{equation*}
where for complex conjugated eigenvalues $(\lambda_{i})_{1,2}=\sigma_i\pm j\omega_i$, $\bC_i$ and $\bI_i$ are defined by $\bC_i=\begin{pmatrix} \sigma_i & \omega_i \\ -\omega_i & \sigma_i\end{pmatrix}$ and $\bI_i=\bI$, respectively, and for real eigenvalues for which $\omega_i=0$, then $\lambda_{i}=\sigma_i$. $\bC_i$ and $\bI_i$ are defined by $\bC_i=\sigma_i$ and $\bI_i=1$, respectively.

\section{Observability}\label{appendixb}

As briefly mentioned in Section \ref{PSSID}, the augmented system (\ref{eq19a}) should be observable. The following proposition describes the observability guarantee of the augmented model under both periodic excitation and non-periodic excitation.
\begin{proposition}\label{propositionappendixb}
If the frequency components of every input signal of the physical system (\ref{eq3}) are the same and meanwhile the system (\ref{eq3}) is observable, the observability of the augmented model (\ref{eq19}) can be checked under the following cases:
\begin{enumerate}[{(1)}]
\item
When the input number is equal to the output number in the system (\ref{eq3}) and if the acceleration sensor is only to measure the output signal, i.e., $\bC_{\rm p}=\bzero$ and $\bC_{\rm v}=\bzero$, then as a result the augmented model (\ref{eq19}) is observable.
\item
When the input number is smaller than the output number in the system (\ref{eq3}), the augmented model (\ref{eq19}) is observable.
\item
When the input number is larger than the output number in the system (\ref{eq3}), the augmented model (\ref{eq19}) is not observable.
\end{enumerate}
\end{proposition}
\begin{proof}
The transfer function of the system (\ref{eq3}) can be described as follows
\begin{align}
\bG_{\rm s}(s)&=\bC_{\rm s}(s\bI-\bA_{\rm s})^{-1}\bB_{\rm s} \nonumber \\
&=\begin{pmatrix}
{\bG_{\rm s}(s)}_{11} & {\bG_{\rm s}(s)}_{12} & \cdots & {\bG_{\rm s}(s)}_{1r} \\
{\bG_{\rm s}(s)}_{21} & {\bG_{\rm s}(s)}_{22} & \cdots & {\bG_{\rm s}(s)}_{2r} \\
\vdots & \vdots & \ddots & \vdots \\
{\bG_{\rm s}(s)}_{m1} & {\bG_{\rm s}(s)}_{m2} & \cdots & {\bG_{\rm s}(s)}_{mr} \\
\end{pmatrix}.\label{B1}
\end{align}

According to (\ref{B1}), it is obvious that if the input number is larger than the output number of the physical system (\ref{eq3}), the matrix $\bG_{\rm s}(s)$ does not have a left inverse, as a result the augmented model (\ref{eq19a}) is not observable, i.e., each input cannot be obtained using the output, while if the input number is smaller than the output number, the augmented model (\ref{eq19}) is observable, i.e., each input can be obtained, as a result the state vector of each input signal model can be obtained using the output signal which can indicate every element of the state vector can be obtained sequentially.

So there remains one case that the input number is equal to the output number of the system (\ref{eq3}). This case depends on the combination of measurement sensors. Here, we only consider about the case under only using the acceleration sensor, i.e., $\bC_{\rm p}=\bzero$ and $\bC_{\rm v}=\bzero$. In the following the observability under this case is studied.

By using the matrix inversion lemma \cite{Verhaegen2007}, we can transform (\ref{B1}) into the following form

\begin{align}
\bG_{\rm s}(s)&=\bC_{\rm s}(s\bI-\bA_{\rm s})^{-1}\bB_{\rm s} \nonumber \\
&=\bC_{\rm s}(\frac{1}{s}\bI+\frac{1}{s^2}\bA_{\rm s}-\frac{1}{s^3}\bA_{\rm s}^2)\bB_{\rm s} \nonumber \\
&=\frac{1}{s}\bC_{\rm s}\bB_{\rm s}+\frac{1}{s^2}\bC_{\rm s}\bA_{\rm s}\bB_{\rm s}-\frac{1}{s^3}\bC_{\rm s}\bA_{\rm s}^2\bB_{\rm s}.\label{B2}
\end{align}

Based on (\ref{B2}), it can be known that if the matrix $\bC_{\rm s}\bB_{\rm s}$ in terms of $\frac{1}{s}$ is full column rank, $\bG_{\rm s}(s)$ will be full column rank. So we only check the column of $\bC_{\rm s}\bB_{\rm s}$ to simplify the problem.

Since
\begin{align}
\bC_{\rm s}\bB_{\rm s}&=\begin{pmatrix} \bC_{\rm p}-\bC_{\rm ac}\bM^{-1}\bK & \bC_{\rm v}-\bC_{\rm ac}\bM^{-1}\bD \end{pmatrix}\begin{pmatrix} \bzero \\ \bM^{-1}\bB \end{pmatrix} \nonumber \\
&=(\bC_{\rm v}-\bC_{\rm ac}\bM^{-1}\bD)\bM^{-1}\bB \nonumber \\
&=-\bC_{\rm ac}\bM^{-1}\bD\bM^{-1}\bB, \label{B3}
\end{align}
it is obvious that the matrix $\bC_{\rm s}\bB_{\rm s}$ is full column rank, so $\bG_{\rm s}(s)$ is full column rank, and as a result the augmented system (\ref{eq19a}) is observable.
\end{proof}

It should be noted that in Proposition \ref{propositionappendixb}, we only consider about the case that $\bC_{\rm p}=\bzero$ and $\bC_{\rm v}=\bzero$.

\begin{remark}
In the above proof process, we do not consider about the properties of the every element of $\bG_{\rm s}(s)$, such as the minimum-phase. However, the method proposed in \cite{Han2018a} can handle the non-minimum phase system, which inspired this paper, so we can only focus on the invertibility of $\bG_{\rm s}(s)$ without considering the minimum-phase property of each element in $\bG_{\rm s}(s)$.
\end{remark}

\bibliographystyle{elsarticle-num}
\bibliography{mybibfile}

\end{document}